\documentclass[a4paper,UKenglish,cleveref, autoref, thm-restate, citecolor=blue]{lipics-v2021}
\usepackage{booktabs}
\usepackage{siunitx}
\usepackage{pgfplots}
\pgfplotsset{compat=1.18}
\usetikzlibrary{patterns}
\usetikzlibrary{positioning, matrix}

\usepackage{wrapfig}

\title{Efficient terabyte-scale text compression via stable local consistency and parallel grammar processing} %TODO Please add

\titlerunning{Efficient terabyte-scale text compression} %TODO optional, please use if title is longer than one line

\author{Diego {D\'iaz-Dom\'inguez}}{University of Helsinki, Finland}{diego.diaz@helsinki.fi}{https://orcid.org/0000-0002-1825-0097}{}%TODO mandatory, please use full name; only 1 author per \author macro; first two parameters are mandatory, other parameters can be empty. Please provide at least the name of the affiliation and the country. The full address is optional. Use additional curly braces to indicate the correct name splitting when the last name consists of multiple name parts.
\authorrunning{D. D\'iaz-Dom\'inguez} %TODO mandatory. First: Use abbreviated first/middle names. Second (only in severe cases): Use first author plus 'et al.'
\Copyright{Diego D\'iaz-Dom\'inguez} %TODO mandatory, please use full first names. LIPIcs license is "CC-BY";  http://creativecommons.org/licenses/by/3.0/
 %TODO mandatory: Please choose ACM 2012 classifications from https://dl.acm.org/ccs/ccs_flat.cfm 

 %TODO mandatory; please add comma-separated list of keywords

\category{} %optional, e.g. invited paper

\relatedversion{} %optional, e.g. full version hosted on arXiv, HAL, or other respository/website
\nolinenumbers %uncomment to disable line numbering

\EventEditors{John Q. Open and Joan R. Access}
\EventNoEds{2}
\EventLongTitle{42nd Conference on Very Important Topics (CVIT 2016)}
\EventShortTitle{CVIT 2016}
\EventAcronym{CVIT}
\EventYear{2016}
\EventDate{December 24--27, 2016}
\EventLocation{Little Whinging, United Kingdom}
\EventLogo{}
\SeriesVolume{42}
\ArticleNo{23}
\def\ATB{\textsf{ATB}}
\def\humans{\textsf{HUM}}
\def\covid{\textsf{COVID}}

\def\kernel{\textsf{KERNEL}}

\def \lcg{\texttt{LCG}}
\def \agc{\texttt{agc}}
\def \sevz{\texttt{zstd}}
\def \rp{\texttt{RePair}}
\def \bigrp{\texttt{BigRePair}}

\begin{document}

\maketitle

\begin{abstract}
We present compression algorithms designed to process terabyte-sized datasets in parallel. Our approach builds on locally consistent grammars, a lightweight form of compression, combined with simple post-processing techniques to achieve further space reductions. Locally consistent grammar algorithms are suitable for scaling as they need minimal satellite information to compact the text, but they are not inherently parallel. To enable parallelisation, we introduce a novel concept that we call \emph{stable} local consistency. A grammar algorithm \textsc{ALG} is stable if for any pattern $P$ occurring in a collection $\mathcal{T}=\{T_1, T_2, \ldots, T_k\}$, instances $\textsc{ALG}(T_1), \textsc{ALG}(T_2), \ldots, \textsc{ALG}(T_k)$ \emph{independently} produce cores for $P$ with the same topology. In a locally consistent grammar, the core of $P$ is a subset of nodes and edges in the parse tree of $\mathcal{T}$ that remains the same in all the occurrences of $P$. This feature enables compression, but it only holds if \textsc{ALG} defines a common set of nonterminal symbols for the strings. Stability removes this restriction, allowing us to run $\textsc{ALG}(T_1), \textsc{ALG}(T_2), \ldots, \textsc{ALG}(T_k)$ in parallel and subsequently merge their grammars into a single output equivalent to that of $\textsc{ALG}(\mathcal{T})$. We implemented our ideas and tested them on massive datasets. Our experiments showed that our method could process 7.9 TB of bacterial genomes in around nine hours, using 16 threads and 0.43 bits/symbol of working memory, achieving a compression ratio of 85x.
\end{abstract}

\section{Introduction}~\label{sec:intro}

Classical dictionary-based compression methods such as Lempel-Ziv (LZ)~\cite{lz76comp,lz77uni} or grammar compression~\cite{Kieffer2000, CLLPPSS05} achieve significant space reductions, but often require extensive resources, limiting their practicality for large datasets. Tools like \texttt{gzip} and \texttt{zstd} provide resource-saving simplifications of LZ that offer acceptable trade-offs for smaller inputs, but still struggle with massive repositories.

Recent heuristics have been developed for large-scale applications. For example, Deorowicz et al.~\cite{d23agc} compress pangenomes by partitioning strings and compressing similar segments together using \texttt{zstd}. Other approaches, Hunt et al.~\cite{atb24}, reorder genomes to improve LZ compression. Grammar algorithms like RePair~\cite{l2000off} and SEQUITUR~\cite{sequitur} achieve high compression ratios, but quickly exceed the available memory as the input grows. Gagie et al.~\cite{g2019rpair} introduced a method using prefix-free parsing~\cite{boucher2019prefix} to scale RePair for large inputs.

Locally consistent grammars~\cite{gawrychowski2018optimal,n2018gr,chri20opt,diaz2021lms,koc24near} is a technique that performs rounds of locally consistent parsing~\cite{mehlhorn1997maintaining, sah94sym, muthu00appr, jez2016really, boucher2019prefix, chri20opt, diaz2021lms} to compress a text $T[1..n]$. This approach recursively segments $T$ based on sequence patterns, producing nearly identical phrases for matching substrings. In the parse tree of a locally consistent grammar, the nodes that cover the occurrences of a pattern $P$ share an area with identical topology and labels. This area is the \emph{core} of $P$~\cite{sah94sym}, and is what makes compression possible. Locally consistent grammars are simple to construct as, unlike LZ or RePair, they only need local information to break $T$ consistently. However, they are not only useful for compression, they also help scaling the processing of large string collections. In fact, they have been used to speed up the computation of the Burrows--Wheeler transform~\cite{diaz2023bwt}, perform pattern matching in grammar-based self-indexes~\cite{chri20opt}, and find maximal exact matches~\cite{diaz2023com,nav24comp}.

Although it is possible to build from $T$ a locally consistent grammar $\mathcal{G}$ of size $O(\delta \log \frac{n\log \sigma}{\delta \log n})$ in $O(n)$ expected time~\cite{koc22toward,koc24near}, $\delta$ being the string complexity~\cite{raskh13sub} and $\sigma$ being the alphabet of $T$, these techniques probably yield less impressive compression ratios in practice than LZ or RePair. However, simple transformations to reduce the size of $\mathcal{G}$ can yield further space reductions. In this regard, Ochoa et al.~\cite{ochoa2018repair} showed that any irreducible grammar can reach the $k$th-order empirical entropy of a string. Their result suggests that building $\mathcal{G}$ and then transforming it might be an efficient alternative to greedy approaches in large datasets.

Parallelising the grammar construction in massive collections is desirable to leverage multi-core architectures. Given an input $\mathcal{T} = \{T_a, T_b\}$, an efficient solution would be to split $\mathcal{T}$ into chunks, say $\{T_a\}$ and $\{T_b\}$, compress the chunks in different instances $\textsc{ALG}(\{T_a\})=\mathcal{G}_a$ and $\textsc{ALG}(\{T_b\})=\mathcal{G}_b$, and merge the resulting (small) grammars $\mathcal{G}_a$ and $\mathcal{G}_b$. However, ensuring the local consistency of the merged grammar is difficult without synchronising the instances. Most locally consistent algorithms assign random fingerprints to the grammar symbols to perform the parsing. However, when there is no synchronisation, different metasymbols emitted by $\textsc{ALG}(\{T_a\})$ and $\textsc{ALG}(\{T_b\})$ that expand to equal sequences of $\mathcal{T}$ could have different fingerprints, thus producing an inconsistent parsing of $T_a$ and $T_b$. Therefore, new locally consistent schemes are necessary to make the parallelisation possible.

\textbf{Our contribution}. We present a parallel grammar compression method that scales to terabytes of data. Our framework consists of two operations. Let $\mathcal{T} = \{T_1, T_2, \ldots, T_k\}$ be a string collection with $\sum_{T_j \in \mathcal{T}} |T_j| =n$ symbols, and let $\mathcal{H}=\{h^{0}, h^{1}, \ldots, h^{l}\}$ be a set of hash functions. The operation $\textsc{BuildGram}(\mathcal{T}, \mathcal{H})=\mathcal{G}$ produces a locally consistent grammar generating strings in $\mathcal{T}$. Furthermore, let $\mathcal{T}_{a}$ and $\mathcal{T}_b$ be two collections with $\textsc{BuildGram}(\mathcal{T}_a, \mathcal{H})=\mathcal{G}_a$ and $\textsc{BuildGram}(\mathcal{T}_b, \mathcal{H})=\mathcal{G}_b$. The operation $\textsc{MergeGrams}(\mathcal{G}_a, \mathcal{G}_b)$ builds a locally consistent grammar $\mathcal{G}_{ab}$ for the collection $\mathcal{T}_{ab}$ that combines $\mathcal{T}_a$ and $\mathcal{T}_b$. \textsc{BuildGram} uses $\mathcal{H}$ to induce a \emph{stable} locally consistent parsing.~The stable property means that $\textsc{BuildGram}(\mathcal{T}_a, \mathcal{H})$ and $\textsc{BuildGram}(\mathcal{T}_b, \mathcal{H})$ \emph{independently} produce cores with the same topology for identical patterns. The set $\mathcal{H}$ assigns random fingerprints to the metasymbols of the grammar under construction to guide the locally consistent parsing, with the fingerprint of a metasymbol $X$ depending on the sequence of its expansion. This feature ensures that metasymbols from different grammars that expand to matching sequences get the same fingerprints. $\textsc{MergeGrams}(\mathcal{G}_a, \mathcal{G}_b)$ leverages the stable property to produce a grammar equivalent to that of $\textsc{BuildGram}(\mathcal{T}_{ab}, \mathcal{H})$, thus allowing us to parallelise the compression. We show that
$\textsc{BuildGram}(\mathcal{T}, \mathcal{H})$ runs in $O(n)$
time w.h.p. and uses $O(G \log G)$ bits of working space, $G$ being the grammar size. Similarly, $\textsc{MergeGrams}(\mathcal{G}_a, \mathcal{G}_b)$ runs in $O(G_a + G_b)$ time and uses $O(G_a\log g_a + G_b\log g_b)$ bits, with $g_a$ and $g_b$ being the number of symbols in $\mathcal{G}_a$ and $\mathcal{G}_b$, respectively. The parsing that we use in \textsc{BuildGram} is inspired by the concept of induced suffix sorting~\cite{n2009li}, which has been shown to be effective for processing strings~\cite{diaz2023bwt,diaz2023com}. In future work, we plan to use our parallel compressor not only to reduce space usage but also to process large inputs. However, we note that the concept of stability is compatible with any locally consistent grammar that uses hashing to break the text. Our experiments showed that our strategy can efficiently compress several terabytes of data. 

\section{Notation and basic concepts}\label{sec:preliminaries}

We consider the RAM model of computation. Given an input of $n$ symbols, we assume our procedures run in random-access memory, where the machine words use $w=\Theta(\log n)$ bits and can be manipulated in constant time. We use the big-$O$ notation to denote time and space complexities (in bits), and the term $\log$ to express the logarithms of base two. 

\subsection{Strings}\label{sec:str}

A string $T[1..n]$ is a sequence of $n$ symbols over an alphabet $\Sigma=\{1,2, \ldots, \sigma\}$. We use $T[j]$ to refer to the $jth$ symbol in $T$ from left to right, and $T[a..b]$ to refer to the substring starting in $T[a]$ and ending in $T[b]$. An equal-symbol run $T[a..b]=s^{c}$ is a substring storing $c$ consecutive copies of $s \in \Sigma$, with $a=1$ or $T[a-1]\neq s$; and $b=n$ or $T[b+1]\neq s$. %We assume by default that runs are of maximum length unless stated otherwise. Moreover, $T{\cdot}T'$ indicates the concatenation of two strings $T$ and $T'$. 

We consider a collection $\mathcal{T}=\{T_1, T_2,\ldots,T_{k}\}$ of $k$ strings as a multiset where each element $T_j \in \mathcal{T}$ has an arbitrary order $j \in [1..k]$. In addition, we use the operator $||\mathcal{T}||=\sum_{T_j \in \mathcal{T}} |T_j|$ to express the total number of symbols. We also use subscripts to differentiate collections (e.g. $\mathcal{T}_a$ and $\mathcal{T}_b$). The expression $\mathcal{T}_{ab} = \mathcal{T}_a \circ \mathcal{T}_b$ denotes the combination of $\mathcal{T}_{a}$ and $\mathcal{T}_b$ into a new collection $\mathcal{T}_{ab}$. We assume that all collections have the same constant-size alphabet $\Sigma$.

\subsection{Grammar compression}\label{sec:grammars}

Grammar compression consists in representing a string $T[1..n] \in \Sigma^{*}$ as a small context-free grammar that generates only $T$~\cite{Kieffer2000,CLLPPSS05}. Formally, a grammar $\mathcal{G}=\{\Sigma, V, \mathcal{R}, S\}$ is a tuple where $\Sigma$ is the alphabet of $T$ (the \emph{terminals}), $V$ is the set of \emph{nonterminals}, $\mathcal{R} \subseteq V \times (\Sigma~\cup~V)^{*}$ is the set of rewriting rules in the form $X \rightarrow Q[1..q]$. The symbol $S \in V$ is the grammar's start symbol. Given two strings $w_a=A{\cdot}X{\cdot}B, w_b=A{\cdot}Q[1..q]{\cdot}B \in (\Sigma\ \cup\ V)^{*}$, $w_b$ rewrites $w_a$ (denoted $w_a \Rightarrow w_b$) if $X \rightarrow Q[1..q]$ exists in $\mathcal{R}$. Furthermore, $w_a$ derives $w_b$, denoted $w_a \Rightarrow^{*} w_b$, if there is a sequence $u_1, u_2,\ldots, u_x$ such that $u_1=w_a$, $u_x=w_b$, and $u_j \Rightarrow u_{j+1}$ for $1\leq j < x$. The string $exp(X) \in \Sigma^{*}$ resulting from $X \Rightarrow^{*} exp(X)$ is the \emph{expansion} of $X$, with the decompression of $T$ expressed as $S \Rightarrow^{*} exp(S)=T$. Compression algorithms ensure that every $X \in V$ occurs only once on the left-hand sides of $\mathcal{R}$. In this way, there is only one possible string $exp(X)$ for each $X$. This type of grammar is referred to as straight-line. The sum of the lengths of the right-hand sides of $\mathcal{R}$ is the grammar size.

The \emph{parse tree} $PT(X)$ of $X \in V$ represents $X \Rightarrow^{*} exp(X)$. Given the rule $X \rightarrow Q[1..q]$, the root of $PT(X)$ is a node $r$ labelled $X$ that has $q$ children, which are labelled from left to right with $Q[1], Q[2], \ldots, Q[q]$, respectively. The $jth$ child of $r$, labelled $Q[j]$, is a leaf if $Q[j] \in \Sigma$; otherwise, it is an internal node whose subtree is recursively defined according to $Q[j]$ and its rule in $\mathcal{R}$.

Post-processing a grammar consists of capturing the remaining repetitions in its rules. For instance, if $XY \in V^{*}$ appears multiple times on the right-hand sides of $\mathcal{R}$, one can create a new rule $Z \rightarrow XY$ and replace the occurrences of $XY$ with $Z$. \emph{Run-length compression} encapsulates each distinct equal-symbol $X^{\ell} \in (\Sigma \cup V)^{*}$ appearing in the right-hand sides of $\mathcal{R}$ as a constant-size rule $X' \rightarrow (X,\ell)$. \emph{Grammar simplification} removes every rule $X \rightarrow Q[1..q]$ whose symbol $X$ appears once on the right-hand sides, replacing its occurrence with $Q[1..q]$.

\subsection{Locally-consistent parsing and grammars}\label{sec:lc_par}

A \emph{parsing} is a partition of a string $T[1..n]$ into a sequence of phrases $T=T[1..j_1-1]T[j_1..j_2-1]\cdots T[j_x..n]$, where the indices $j_1<j_2< \ldots < j_x$ are \emph{breaks}. Let $par(o,o')=T[j_{y}..j_{y+1}-1]T[j_{y+1}..j_{y+2}-1] \ldots T[j_{y+u-1}..j_{y+u}-1]$ denote the $u$ phrases that cover a substring $T[o..o']$, with $o \in [j_{y}..j_{y+1})$ and $o' \in [j_{y+u-1},j_{y+u})$. A parsing is \emph{locally consistent}~\cite{cole1986deterministic} iff, for any pair of equal substrings $T[a..b]=T[a'..b']$, $par(a,b)$ and $par(a',b')$ differ in $O(1)$ phrases at the beginning and $O(1)$ at the end, with their internal phrase sequences identical.

%differ in the first $O(1)$ and the last $O(1)$ breaks, with the internal breaks being the same.
%have at most $O(1)$ differences. One can achieve local consistency by designing a parsing algorithm that utilises $T$'s sequence to define the breaks. In this way, $T[a..b]$ and $T[a'..b']$ yield almost identical phrase sequences that differ only in the flanks where the context changes. 

\begin{figure}[t]
\centering
\includegraphics[width=0.7\textwidth]{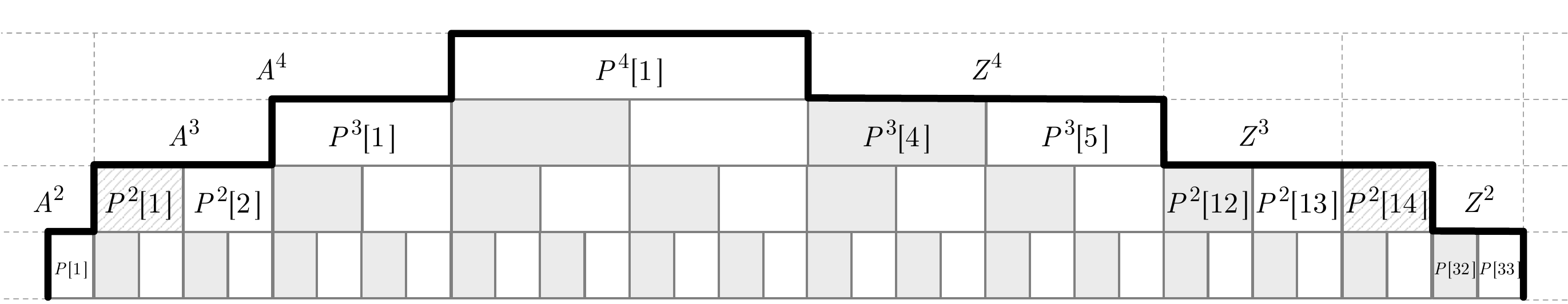}
\caption{Locally consistent grammar compression of $P[1..33]$. The first row (bottom-up) is $P$, and the next rows are the metasymbols for parsing rounds. The grey boxes are breaks. The boxes below the thick black are the core of $P$. The $A^{i}$ and $Z^{i}$ change if the context of $P$ changes.}
\label{fig:core_examp}
\end{figure}

A locally consistent parsing scheme relevant to this work is that of Nong et al.~\cite{n2018gr}. They originally described their idea to perform induced suffix sorting, but it has been shown that it is also locally consistent~\cite {diaz2021lms,deng2022fm}. They define a \emph{type} for each position $T[\ell]$:

\begin{equation}\label{eq:lex_types}
   \begin{array}{ll}
    \text{L-type} \iff T[\ell] > T[\ell+1]\text{ or }T[\ell] = T[\ell+1]\text{ and }T[\ell+1]\text{ is L-type}.\\
    \text{S-type} \iff T[\ell] < T[\ell+1]\text{ or }T[\ell] = T[\ell+1]\text{ and }T[\ell+1]\text{ is S-type}.\\
    \text{LMS-type} \iff T[\ell]\text{ is S-type and }T[\ell-1]\text{ is L-type}.
   \end{array} 
\end{equation}

Their method scans $T$ from right to left and sets a break at each LMS-type position. 

One can create a grammar $\mathcal{G}=\{\Sigma, V, \mathcal{R}, S\}$ that only generates $T$ by applying successive rounds of locally consistent parsing. Different works present this idea slightly differently (see~\cite{sah94sym, muthu00appr, jez2016really, boucher2019prefix, chri20opt, diaz2021lms}), but the procedure is similar: in every round $i$, the algorithm receives a string $T^{i}$ as input ($T^i=T$ with $i=1$) and performs the following steps:

\begin{enumerate}
\item Break $T^{i}$ into phrases using a locally consistent parsing.\label{step_one}
\item Assign a nonterminal $X$ to each distinct sequence $Q[1..q]$ that is a phrase in $T^{i}$.\label{step_two}
\item Store every $X$ in $V$ and its associated rule $X \rightarrow Q[1..q]$ in $\mathcal{R}$.\label{step_three}
\item Replace the phrases in $T^{i}$ by their nonterminals to form the string $T^{i+1}$ for $i+1$.\label{step_four}
\end{enumerate}

The process ends when $T^{i}$ does not have more breaks, in which case it creates the rule $S \rightarrow T^{i}$ for the start symbol and returns the resulting grammar $\mathcal{G}$. The phrases have a length of at least two, so the length of $T^{i+1}$ is half the length of $T^{i}$ in the worst case. Consequently, the algorithm incurs in $O(\log n)$ rounds and runs in $O(n)$ time.%, assuming the algorithm computes the breaks in linear time.

The output grammar $\mathcal{G}$ is locally consistent because it compresses the occurrences of a pattern $P[1..m]$ largely in the same way. The first parsing round transforms the occurrences into substrings in the form $A^{2}[1..a^{2}]{\cdot}P^{2}[1..m^{2}]{\cdot}Z^{2}[1..z^{2}] \in V^{*}$, where the superscripts indicate symbols in $T^2$. The blocks $A^{2}[1..a^{2}]$ and  $Z^{2}[1..z^{2}]$ have $O(1)$ variable nonterminals that change with $P$'s context, while $P^{2}[1..m^{2}]$ remains the same in all occurrences. In the second round, $P^{2}[1..m^{2}]$ yields strings in the form $A^{3}[1..a^{3}]{\cdot}P^{3}[1..m^{3}]{\cdot}Z^{3}[1..z^{3}]$ that recursively have the same structure. The substring $P^{i}[1..m^{i}]$ remains non-empty during the first $O(\log m)$ rounds. The substrings $P^{i}$, with $i=1,\ldots, O(\log m)$, conform the \emph{core} of $P$~\cite{sah94sym} (see Figure~\ref{fig:core_examp}).

\subsection{Hashing and string fingerprints}\label{sec:hs_and_fp}

\emph{Hashing} refers to the idea of using a function $h: \mathcal{U} \rightarrow [m]$ to map elements in a universe $\mathcal{U}$ to integers in a range $[m]=\{0,1,\ldots,m-1\}$ uniformly at random.
%The function is \emph{collision-free} if, for any pair $x,y \in \mathcal{U}$, their hash values (or fingerprints) are different $h(x) \neq h(y)$. Let $[\sigma]=\{1,2,\ldots, \sigma\}$ be the universe, $p>\sigma$ a prime number, and $a \in [p]_{+}=\{1,2,\ldots,p-1\}$ and $b \in [p]=\{0, 1, \ldots, p-1\}$ integers chosen uniformly at random. The formula  $h(x) = (ax + b) \bmod p$ is a collision-free function $h: \mathcal{U} \rightarrow [p]$.
When the universe $\mathcal{U}$ is large and only an unknown subset $\mathcal{K} \subset \mathcal{U}$ requires fingerprints over a range $[m]$ with $|\mathcal{K}| < m \ll |\mathcal{U}|$, the typical solution is to use a \emph{universal} hash function. For any pair $x,y \in \mathcal{U}$, a universal hash function ensures that the collision probability is $\text{Pr}[h(x)=h(y)]=1/m$. Let $\mathcal{U} =\{0, 1,2,\ldots, \sigma\}$ be the universe, $p>\sigma$ a prime number, and $a \in [p]_{+}=\{1,2,\ldots,p-1\}$ and $b \in [p]=\{0,1, \ldots, p-1\}$ integers chosen uniformly at random. One can make $h : \mathcal{U} \rightarrow [m]$ universal using the formula $h(x) = ((ax + b) \bmod p) \bmod m$, with $p>m$.

These ideas can be adapted to produce fingerprints for a set $\mathcal{K} \subset \Sigma^{*}$ of strings over an alphabet $\Sigma=\{1,2,\ldots, \sigma\}$~\cite{dietz92pol}. Pick a prime number $p > \sigma$ and choose an integer $c\in [p]_{+}$ uniformly at random. Then, build the degree $q$ polynomial $h(Q[1..q]) = \left(\sum_{i=1}^{q} Q[i]\cdot c^{i-1}\right) \bmod p$ over $[p]$ and regard the symbols $Q[1], Q[2], \ldots, Q[q]$ in $Q[1..q] \in \mathcal{U}$ as the polynomial's coefficients. Additionally, compose $h$ with a universal hash function $h': [p] \rightarrow [m]$ to obtain a fingerprint in $[m]$. Let $a,b,c \in [p]_{+}$ be three integers chosen uniformly at random. Then the function becomes $h'(Q[1..q]) = \left(\left( a \left(\sum_{i=1}^{q} Q[i]\cdot c ^{i-1}\right) + b \right) \bmod p\right) \bmod m.$
%\begin{equation}
%\end{equation}

\section{Our methods}

\subsection{The grammar model}\label{sec:gram_feat}

We first introduce the features of the locally consistent grammar $\mathcal{G}=\{\Sigma, V, \mathcal{R}, S\}$ we build with \textsc{BuildGram} and \textsc{MergeGrams}. Given a rule $X \rightarrow Q[1..q] \in \mathcal{R}$, the operator $rhs(X)$ is an alias for the string $Q[1..q]$, with $rhs(s)=s$ when $s \in \Sigma$. The \emph{level} of $X$ is the number of edges in the path of $PT(X)$ that starts at its root and ends in a leaf. $\mathcal{G}$ is \emph{fully balanced} because, for any $X \in V \setminus \{S\}$, every root-to-leaf path in $PT(X)$ has the same number of edges. We extend the use of ``level'' for the rule $X \rightarrow Q[1..q] \in \mathcal{R}$ associated with $X$. We indistinctly use the terms nonterminal and metasymbol to refer to $X\in V$.

The start symbol $S$ is associated with a rule $S \rightarrow C[1..k]$ where $C[1..k]$ is a sequence such that each $exp(C[j])=T_j[1..n_j] \in \mathcal{T}$ and $PT(C[j])$ has height $O(\log n_j)$. The symbols in $C[1..k]$ can have different heights and, hence, different levels (this idea will become clear in Section~\ref{sec:gram_algo}). Let $l_{max}$ be the highest level among the elements in $C[1..k]$. We set the level of $S$ equal to $l=l_{max}+1$, which we regard as the height of $\mathcal{G}$.
%The grammar $\mathcal{G}$ we build has a start symbol $S$ with multiple replacements $S \rightarrow X_1\ |\ X_2\ |\ldots|\ X_k$, where each $exp(X_x)=T_x[1..n_x] \in \mathcal{T}$ and $PT(X_x)$ has height $O(\log n_x)$. We remark that $S$ is the only exception to the straight-line condition as any $X \in V \setminus \{S\}$ has one replacement in $\mathcal{R}$.
%is not straight-line (Section~\ref{sec:grammars}) as the start symbol $S$ 
%This situation occurs because our grammar algorithm independently compresses every $T_x[1..n_x] \in \mathcal{T}$ to a nonterminal $X_x \in V$,  We reserve $C[1..k]$ to refer exclusively to the concatenated compressed strings. 

We define the partitions $\mathcal{R} = \{\mathcal{R}^{1}, \ldots, \mathcal{R}^{l-1}\}$ and $V =\{V^{1}, V^{2}, \ldots, V^{l-1}\}$, where every pair $(\mathcal{R}^{i}, V^{i})$ is the set of rules and nonterminals (respectively) with level $i \in [1..l-1]$. In each subset $\mathcal{R}^{i}$, the left-hand sides are symbols over the alphabet $V^{i}$, while the right-hand sides are strings over $V^{i-1}$. Further, we consider $V^{0}=\Sigma$ to be the set of terminals. %The value $g^{i}=|V^{i}|$ is the number of nonterminals with level $i$ and $G^{i}=\sum_{X \in V^{i}} |rhs(X)|$ is the size of level $i$.

\subsection{Fingerprints for the grammar symbols}\label{sec:str_hash}

In this section, we describe the set $\mathcal{H}$ of hash functions that assign fingerprints in \textsc{BuildGram}. %We first create a function $h^{0}$ that maps terminal symbols to integers in an arbitrary integer range. For this purpose, we pick random prime number $p_{0} > \sigma$ and two random integers $a_0 \in [p_0]_{+}=\{1,2,\ldots, p_{0}-1\}$ and $b_0 \in [p_{0}]=\{0,1,\dots, p_{0}-1\}$, and define a collision-free function $h^{0}: \Sigma \rightarrow [p_0]$ (Section~\ref{sec:hs_and_fp}).
The universal hash function $h^{0} : \Sigma \rightarrow [m_{0}]$ maps terminal symbols to integers over an arbitrary range $[0,1,\ldots, m_{0}-1]$, with $m_0>\sigma$. Furthermore, each function $h^{i}$, with $1 \leq i< l$, recursively assigns fingerprints to the right-hand sides of $\mathcal{R}^{i}$. Let $[m_{i-1}]$ be the integer range for the fingerprints emitted by $h^{i-1}$. We choose a random prime number $p_{i}>m_{i-1}$, three integer values $a_i, b_i, c_i \in [p_i]_{+}$, and a new integer $m_{i}$. Now, given a rule $X \rightarrow Q[1..q] \in \mathcal{R}^{i}$, we compute the fingerprint for $Q[1..q]$ as   

\begin{equation}\label{eq:lvl_hash}
   h^{i}(Q[1..q]) = \left(\left(a_i \left(\sum_{j=1}^{q} h^{i-1}(rhs(Q[j]))\cdot c^{j-1}_i\right) + b_i \right) \bmod p_i\right) \bmod m_i.
\end{equation}

%The function $h^{i}(Q[1..q])$ computes a fingerprint for $Q[1..q]$ using as input the sequence of fingerprints in $h^{i-1}$ for $rhs(Q[1]), rhs(Q[2]), \ldots, rhs(Q[q])$.

Although $h^{i} : [m^{i-1}]^{*} \rightarrow [m^{i}]$ computes a fingerprint for a string, we associate this fingerprint with $X \in V^{i}$ because each $Q[1..q]$ has one possible $X$. Notice that the recursive definition of $h^{i}$ implicitly traverses $PT(X)$ and ignores the nonterminal labelling $PT(X)$. As a result, the value that $h^{i}$ assigns to $Q[1..q]$ (or equivalently, to $X$) depends on $exp(X)$, the functions $h^{0},h^{1}, \ldots, h^{i}$, and the topology of $PT(X)$. %We will use this feature in Section~\ref{sec:est_prop} to define the concept of stable local consistency. 
%As mentioned,
In practice, we avoid traversing $PT(X)$ by operating bottom-up over $\mathcal{G}$: When we process $\mathcal{R}^{i}$, the fingerprints in $h^{i-1}$ for $V^{i-1}$ that we require to obtain the fingerprints in $h^{i}$ are available.

\textsc{BuildGram} does not know a priori the number of hash functions $\mathcal{H}$ needs to compress an input. However, the locally consistent grammar algorithm that we use requires $O(\log n)$ rounds of parsing to compress $T[1..n]$ (Section~\ref{sec:gram_algo}). If we consider the number of rounds plus the function $h^{0}$ for the terminals, then $|\mathcal{H}| \geq \lceil \log n \rceil + 1$ is enough to process $T[1..n]$.

%\textbf{Technical note}. Our technique requires a set $\mathcal{H}$ with, at least, as many functions as the height of the tallest grammar we could ever manipulate. Nevertheless, knowing this information beforehand is impossible because it depends on the (unknown) length of the longest string we process at any time. We solve this issue by regarding $\mathcal{H}$ as a semi-static set: we can not remove or modify functions already present, but we can append new ones. Thus, if we run out of functions in $\mathcal{H}$ while creating a new grammar, we append new ones to the set. In any case, each $h^{i}$ requires $O(w)$ bits of space, and appending, say, $40$ random functions uniformly at random should work for strings of length up to $2^{40}-1$. This limit is enough for most practical scenarios. Another important aspect is the choice of $m_i$ for each $V^{i}$. Ideally, $m_i$ should always be greater than $|V^{i}|$ to keep a low fingerprint collision probability. In our experiments (Section~\ref{sec:exp}), $|V^{i}|$ grew sublinearly and never exceeded $2^{32}$, even when the input size was several TBs. Therefore, values for $m_i$ in the range $[2^{32}..2^{40}]$ should work fine for massive datasets. 

\subsection{Our grammar algorithm}\label{sec:gram_algo}

Our procedure $\textsc{BuildGram}(\mathcal{T}, \mathcal{H})$ receives as input a collection $\mathcal{T}=\{T_1, T_2, \ldots, T_k\}$ of $k$ strings and a set $\mathcal{H}$ of hash functions, and returns a locally consistent grammar $\mathcal{G}=\{\Sigma, V, \mathcal{R}, S\}$ that only generates strings in $\mathcal{T}$. We assume $\mathcal{H}$ has at least $\lceil \log n_{max} \rceil +1$ elements (see Section~\ref{sec:str_hash}), where $n_{max}$ is the length of the longest string in $\mathcal{T}$.

The algorithm of \textsc{BuildGram} is inspired by the parsing of Nong et al.~\cite{Li2010a}, which has been used not only for compression~\cite{n2018gr}, but also for the processing of strings~\cite{diaz2023bwt} (see Section~\ref{sec:speed_up}). However, as noted in the Introduction, the ideas we present here and in the next sections are compatible with any locally consistent grammar that uses hashing.

\textbf{Overview of the algorithm}. \textsc{BuildGram} constructs $\mathcal{G}$ in successive rounds of locally consistent parsing. In each round $i$, we run steps~\ref{step_one}-\ref{step_four} of Section~\ref{sec:lc_par}, breaking the strings of $\mathcal{T}$ individually, but collapsing the rules in the same grammar $\mathcal{G}$. When we finish round $i$, we flag each string $T^{i+1}_j \in \mathcal{T}^{i+1}$ with length one as inactive (that is, fully compressed) and stop the compression in round $i+1$ if there are no active strings. Subsequently, we create the sequence $C[1..k]$ with compressed strings (that is, symbols we marked as inactive), create the start symbol $S \in V$ with the corresponding rule $S \rightarrow C[1..k] \in \mathcal{R}$, and finish \textsc{BuildGram}. 

\subsubsection{Parsing mechanism}

We parse the active strings of $\mathcal{T}^{i}$ (step~\ref{step_one} of the round) using a variant of the parsing of Nong et al.~\cite{n2013pr} (Section~\ref{sec:lc_par}) that employs Equation~\ref{eq:lvl_hash} to randomise the sequences that trigger breaks. We refer to this modification as \textsc{RandLMSPar}.

Let $X, Y \in V^{i-1}$ be any pair of nonterminals. We define the partial order $\prec$ as follows:

\[
X \prec Y \iff 
\begin{cases} 
h^{i-1}(rhs(X)) < h^{i-1}(rhs(Y)) & \text{if } h^{i-1}(rhs(X)) \neq h^{i-1}(rhs(Y)), \\
\text{undefined} & \text{if } h^{i-1}(rhs(X)) = h^{i-1}(rhs(Y)).
\end{cases}
\]

Additionally, we define the equivalence relation:
\[
X \sim Y \iff h^{i-1}(rhs(X)) = h^{i-1}(rhs(Y)),
\]

to cover the cases where $X=Y$ or $X\neq Y$ and their fingerprints collide. Now, let $T^{i}_{j}[\ell], T^{i}_j[\ell+1] \in V^{i-1}$ be two adjacent positions in some string $T^{i}_j \in \mathcal{T}^{i}$ during round $i$. We redefine the types of Equation~\ref{eq:lex_types} for $T^{i}_{j}[\ell]$ as follows:

\begin{equation}\label{eq:rand_types}
   \begin{array}{ll}
    \text{L-type} \iff T^{i}_j[\ell] \succ T^{i}_j[\ell+1]\text{ or }T^{i}_j[\ell] \sim T^{i}_j[\ell+1]\text{ and }T^{i}_j[\ell+1]\text{ is L-type}.\\
    \text{S-type} \iff T^{i}_j[\ell] \prec T^{i}_{j}[\ell+1]\text{ or }T^{i}_{j}[\ell] \sim T^{i}_{j}[\ell+1]\text{ and }T^{i}_{j}[\ell+1]\text{ is S-type}.\\
    \text{LMS-type} \iff T^{i}_j[\ell]\text{ is S-type and }T^{i}_{j}[\ell-1]\text{ is L-type}.
   \end{array} 
\end{equation}

The above types are undefined for the suffix $T^{i}_{j}[\ell..n_{j}]=s^c$ that is an equal-symbol run. This restriction implies that $T^{i}_{j}[\ell..n_{j}]$ cannot have LMS-type positions (that is, breaks). %However, this aspect does not affect \textsc{BuildGram}. 

A substring $T^{i}_{j}[\ell..r]$ is a phrase in \textsc{RandLMSPar} if the following two conditions hold: (i) $\ell=1$ or $T^{i}_{j}[\ell]$ is LMS-type; and (ii) $r=n_{j}$ or $T^{i}_{j}[r+1]$ is LMS-type.

\begin{figure}[!t]
\centering
\includegraphics[width=\textwidth]{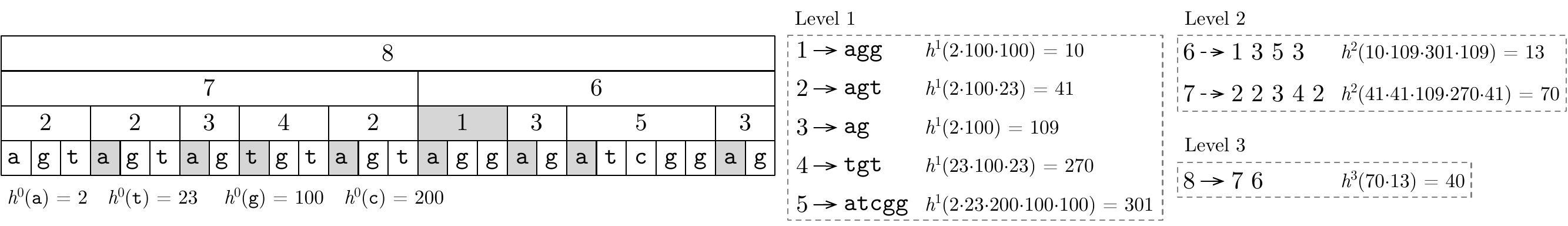}
\caption{Example of \textsc{BuildGram} with the input string \texttt{agtagtagtgtagtaggagatcggag} and the hash functions $\mathcal{H}=\{h^{0}, h^{1}, h^{2}, h^{3}\}$. The grey boxes indicate the breaks induced by $\mathcal{H}$.}
\label{fig:build_gram}
\end{figure}

Once we partition the strings of $\mathcal{T}^{i}$ and store its distinct phrases in a set $\mathcal{S}$, we assign metasymbols to the phrases (step~\ref{step_two} of the round). Let $g^{i-1}=|\Sigma \cup V|$ be the number of symbols $\mathcal{G}$ has when we begin round $i$. We assign the nonterminal $X = g^{i-1} + o \in V^{i}$ to the $oth$ string $Q[1..q] \in \mathcal{S}$ and add $X \rightarrow Q[1..q] \in \mathcal{R}^{i}$. We note that the order of the strings in $\mathcal{S}$ is arbitrary and does not affect the properties of our method. The last step of the parsing round is to create $\mathcal{T}^{i+1}$ by replacing the phrases in $\mathcal{T}^{i}$ with their nonterminals in $V^{i}$. Figure~\ref{fig:build_gram} shows an example of \textsc{BuildGram}.

Our hash-based parsing mechanism induces a property in the grammar that we call $\emph{stable local consistency}$. %Put simply, it means that, for any pattern $P$, its occurrences will be largely parsed in the same way in \emph{all} the texts where it occurs.

\begin{definition}\label{def:est}
    Stable local consistency: Let {\normalfont \textsc{ALG}} be an algorithm that produces a locally consistent grammar. Additionally, let $P[1..m]$ be a pattern occurring in an arbitrary number of text collections. {\normalfont \textsc{ALG}} is stable iff, for any pair of distinct texts $\mathcal{T}_a\neq \mathcal{T}_b$, the instances {\normalfont $\textsc{ALG}(\mathcal{T}_a)=\mathcal{G}_a$} and {\normalfont $\textsc{ALG}(\mathcal{T}_b)=\mathcal{G}_b$} independently produce a core for $P$ {\normalfont (Section~\ref{sec:lc_par})} with identical tree topology and different nonterminal labels. The term ``independently'' means that {\normalfont $\textsc{ALG}(\mathcal{T}_a)$} does not use information about $\mathcal{G}_b$ in its execution, and vice versa.
\end{definition}

%The argument is simple: The classification of a position $T^{i}[\ell]$ as a break depends on the fingerprint in $h^{i-1}$ for $T^{i}_j[\ell]$. As we stated in Section
The classification of $T^{i}_j[\ell]$ as a break depends on the fingerprint resulting from the evaluation of $exp(T^{i}_j[\ell])$ with functions $h^{0},h^{1},\ldots, h^{i-1}$. Thus, if $exp(T^{i}[\ell])$ appears in another collection $\mathcal{T}' \neq \mathcal{T}$, surrounded by an identical context, processing $\mathcal{T}'$ with $\mathcal{H}$ produces breaks and a core topology for $exp(T^{i}_{j}[\ell])$ identical to that of $\mathcal{T}$. The stable property depends on the use of $\mathcal{H}$ and not on the parsing algorithm, which means that any locally consistent parsing compatible with hashing (e.g. \cite{mehlhorn1997maintaining,chri20opt,gawrychowski2018optimal}) would achieve similar results. %However, we use \textsc{RandLMSPar} only because it is convenient for processing strings (see Section~\ref{sec:speed_up}). 

%This idea allows us to deal with adversarial inputs where the nonrandomised parsing of Nong et al. does not perform well.
%Definition~\ref{def:est} holds as long as we process $\mathcal{T}_a$ and $ \mathcal{T}_b$ with the same set $\mathcal{H}$. The result is undefined if, for example, we build $\textsc{BuildGram}(\mathcal{T}_a, \mathcal{H}_a)=\mathcal{G}_a$, and $\textsc{BuildGram}(\mathcal{T}_b, \mathcal{H}_b)=\mathcal{G}_b$, with $\mathcal{H}_a \neq \mathcal{H}_b$. The stable property allows us to efficiently merge $\mathcal{G}_a$ and $\mathcal{G}_b$ to produce a locally consistent grammar $\mathcal{G}_{ab}$ that generates strings in $\mathcal{T}_{ab}=\mathcal{T}_a \circ \mathcal{T}_b$. %In principle, this task would require any pattern $P$ occurring in both $\mathcal{T}_a$ and $\mathcal{T}_b$ to have the same core, but this property partially holds due to the stable condition of $\mathcal{G}_a$ and $\mathcal{G}_b$. Therefore, we have to identify the core of $P$ in $\mathcal{G}_a$ and $\mathcal{G}_b$ and assign a common set of values to their nonterminals. 

\subsection{Implementing our grammar algorithm}~\label{sec:gram_algo_imp}

%This section describes the practical aspects of \textsc{BuildGram}.
Calculating the LMS-type positions of \textsc{RandLMSPar} during the round $i$ requires Equation~\ref{eq:rand_types} to obtain the type of each $T^{i}_{j}[\ell]$, which involves knowing the relative $\prec$ order of $T^{i}_{j}[\ell]$ and $T^{i}_{j}[\ell+1]$. We obtain this information by feeding $rhs(T^{i}_j[j])$ and $rhs(T^{i}_{j}[\ell+1])$ to Equation~\ref{eq:lvl_hash}. The problem is that Equation~\ref{eq:lvl_hash} decompresses $T^{i}_{j}[\ell]$ and $T^{i}_{j}[\ell+1]$ from $\mathcal{G}$, adding a logarithmic penalty on the grammar construction. We avoid decompression by keeping an array $F[1..|\Sigma \cup V|]$ that stores the fingerprints of the symbols we already have in $\Sigma \cup V$.

At the beginning of \textsc{BuildGram}, we initialize $F$ with $\Sigma$ elements, where every $F[s]=h^{0}(s)$ stores the fingerprint of $s \in \Sigma$. Then, each round $i$ keeps in $F[T^{i}[\ell]]$ the fingerprint $h^{i-1}(rhs(T^{i}[\ell]))$. After we finish round $i$, we store the fingerprint $F[X]=h^{i}(rhs(X))$ of every $X \in V^{i}$ so that we can compute the types of $\mathcal{T}^{i+1}$ in the next round $i+1$. %Thus, for every $X \in V^{i}$, we store . %We use the precomputed values in $F$ to avoid decompressing $X$ during the execution of $h^{i}(rhs(X))$.

Let $rhs(X)= Q[1..q]$ be the replacement for $X$. We modify Equation~\ref{eq:lvl_hash} as follows:

\begin{equation}\label{eq:fast_hash}
   F[X] = h^{i}(Q[1..q], F) = \left(\left(a_i \left(\sum_{j=1}^{q} F[Q[j]]\cdot c^{j-1}_i\right) + b_i \right) \bmod p_i\right) \bmod m_i.
\end{equation}

This operation is valid because the alphabet of $Q[1..q]$ is $V^{i-1}$, and $F$ already has its fingerprints.

\textbf{A note on collisions}. The consecutive positions $T^{i}_j[\ell] \neq T^{i}_j[\ell+1]$ with $h^{i-1}(rhs(T^{i}[\ell]))=h^{i}(rhs(T^{i}[\ell+1]))$ (i.e. collisions) never cause a break because $T^{i}[\ell] \sim T^{i}[\ell+1]$. Intuitively, the more contiguous colliding symbols we have, the less breaks the parsing produces, and the more inefficient the compression becomes. The chances of this situation are small if the hash function $h^{i-1} \in \mathcal{H}$ emits fingerprints in the range $[m^{i-1}]$ with $|V^{i-1}|<m^{i-1}$. However, we do not know a priori $|V^{i-1}|$, so we have to choose a very large $m^{i-1}$. This decision has a trade-off because a large $m^{i-1}$ means that cells of $F$ require more bits, and hence more working memory. In Section~\ref{sec:exp}, we investigate suitable values for $m^{i}$ in large collections.

Now, we present the theoretical cost of \textsc{BuildGram}.

\begin{theorem}
Let $\mathcal{T}$ be a collection of $k$ strings and $||\mathcal{T}||=n$ symbols, where the longest string has length $n_{max}$. Additionally, let $\mathcal{H}$ be a set of hash functions with $|\mathcal{H}| \geq \lceil \log n_{max} \rceil + 1$ elements. $\textsc{BuildGram}(\mathcal{T}, \mathcal{H})=\mathcal{G}$ runs in $O(n)$ time w.h.p. and requires $O(G\log w)$ bits on top of $\mathcal{T}$, where $G$ is the grammar size of $\mathcal{G}$. 
\end{theorem}

\begin{proof}
Calculating the type of each $T^{i}_j[\ell]$ in $T^{i}_{j}[1..n_j] \in \mathcal{T}^{i}$ takes $O(n_j)$ time if we have the array $F$ with precomputed fingerprints of $V^{i-1}$. In addition, we can use a hash table to record the parsing phrases in $T^{i}_j$, which takes $O(n_x)$ time w.h.p. If we consider all the strings in $\mathcal{T}^{i}$, the running time of the parsing round $i$ is $\mathcal{O}(||\mathcal{T}^{i}||)$ in expectation. On the other hand, all the phrases in $T^{i}_j$ have length $>1$, except (possibly) one phrase at each end of $T^{i}_j$. Therefore, $\mathcal{T}^{i+1}$ has $\frac{n}{2^{i}}+2k$ symbols in the worst case.  
%, which means that $|\mathcal{T}^{i+1}|$ has half the number of symbols than $\mathcal{T}^{i}$, but add $2$ extra symbols per string in the worst case. 
%T^{i+1}_j$ has length at most $n_j/2+ 2$, and that $\mathcal{T}^{i+1}$ has size at most $ \sum_{j=1}^{k} \frac{n_j}{2} + 2 = \frac{n}{2^{i}} + 2k$.
Considering that \textsc{BuildGram} requires $O(\log n_{max})$ rounds, the cumulative length of $\mathcal{T}^{1}, \mathcal{T}^{2}, \ldots, \mathcal{T}^{\log n_{max}}$ is at most $n + \Sigma^{\log n_{max}}_{i=1} \frac{n}{2^{i}} + 2k \leq 2n + 2k \log n_{max}$, with $2k\log n_{max}$ being the contribution of the phrases with length one. However, \textsc{BuildGram} stops processing a string $T_j \in \mathcal{T}$ as soon as it is fully compressed, meaning that length-one phrases contribute $\sum_{T_j \in \mathcal{T}} \log |T_j| \leq k\log n_{max}$ elements in the worst case. Therefore, as $\sum_{T_j \in \mathcal{T}} \log |T_j|< \sum_{T_j \in \mathcal{T}} |T_j|=n$, the running time of \textsc{BuildGram} is $O(n)$ w.h.p.

Let $g=|\Sigma| + |V|$ be the number symbols in $\mathcal{G}$. The $O(G\log G) + g\log w + |\mathcal{H}|w = O(G\log w)$ bits of working space in \textsc{BuildGrams} represent the $O(G\log G)$ bits of the hash tables, the array $F$ that stores the $g$ fingerprints, and the hash functions in $\mathcal{H}$.
\end{proof}

\subsection{Merging grammars}\label{sec:mer_gram}

We now present our algorithm for merging grammars. Let $\mathcal{T}_{a}$ and $\mathcal{T}_{b}$ be two collections, with $n_a$ and $n_b$ being the lengths of the longest strings in $\mathcal{T}_{a}$ and $\mathcal{T}_{b}$, respectively, and $\mathcal{T}_{ab} =  \mathcal{T}_a \circ \mathcal{T}_b$ being their union (Section~\ref{sec:str}). Furthermore, let $\mathcal{G}_a=\textsc{BuildGram}(\mathcal{T}_a, \mathcal{H})$ and $\mathcal{G}_b=\textsc{BuildGram}(\mathcal{T}_b, \mathcal{H})$ be grammars that
only generate strings in $\mathcal{T}_a$ and $\mathcal{T}_b$, respectively. We assume that $\mathcal{H}$ has $|\mathcal{H}| \geq \lceil \log \max(n_a, n_b) \rceil + 1$ elements (see Section~\ref{sec:str_hash}). The instance $\textsc{MergeGrams}(\mathcal{G}_a, \mathcal{G}_a)$ returns a locally consistent grammar $\mathcal{G}_{ab}$ that only generates strings in $\mathcal{T}_{ab}$, and that is equivalent to the output of $\textsc{BuildGram}(\mathcal{T}_{ab}, \mathcal{H})$.

%We introduce a minor modification to the definitions in Section~\ref{sec:gram_feat} to explain our algorithm. Specifically, the subscript $a$ (respectively, $b$) indicates that a grammar component comes from $\mathcal{G}_a$ (respectively, $\mathcal{G}_b$). For instance, $l_a$ is the height of $\mathcal{G}_a$, and $\mathcal{R}_a = \{\mathcal{R}_a^{i}, \mathcal{R}_a^{2}, \ldots, \mathcal{R}^{l_a-1}_a\}$ is the partition of the rules in $\mathcal{G}_a$ according to their levels. The same idea applies for components with the subscript $b$. 

\textbf{Overview of the Algorithm}. The merging consists of making $\mathcal{G}_a$ absorb the content that is unique to $\mathcal{G}_b$. Specifically, we discard rules from $\mathcal{G}_b$ whose expansions occur in $\mathcal{T}_a$, and for those expanding to sequences not in $\mathcal{T}_a$, we add them as new rules in $\mathcal{G}_a$. %Finally, we concatenate the compressed strings $C_a[1..k_a]$ and $C_b[1..k_b]$ as one single sequence $C_{ab}=C_a[1..k_a]{\cdot}C[1..k_b]$.

\subsection{The merge algorithm}\label{sec:mg_imp}

%We will introduce some notation before explaining \textsc{MergeGrams}.
For the grammar $\mathcal{G}_a$, let $\mathcal{R}_{a}$ be its set of rules, let $V_a$ be its set of nonterminals, and let $l_a$ be its height. The symbols $\mathcal{R}_b$, $V_b$, and $l_b$ denote equivalent information for $\mathcal{G}_b$. We consider the partitions $\mathcal{R}_{a} = \{\mathcal{R}^{1}_{a}, \mathcal{R}^{2}_{a}, \ldots, \mathcal{R}^{l_a-1}_{a}\}$ and $V_{a} =\{V^{1}_{a}, V^{2}_{a}, \ldots, V^{l_a-1}_{a}\}$, where every $(\mathcal{R}^{i}_{a}, V^{i}_{a})$ is the set of rules and nonterminals (respectively) that $\textsc{BuildGram}(\mathcal{T}_a, \mathcal{H})$ produced during the parsing round $i \in [1..l_a-1]$. The elements $\mathcal{R}_{b} = \{\mathcal{R}^{1}_{b}, \mathcal{R}^{2}_{b}, \ldots, \mathcal{R}^{l_b-1}_{b}\}$ and $V_{b} =\{V^{1}_{b}, V^{2}_{b}, \ldots, V^{l_b-1}_{b}\}$ are equivalent partitions for $\mathcal{G}_{b}$. %Notice that $l_a$ and $l_b$ can be different but $l_a,l_b < |\mathcal{H}|+1$.
%For each $(V^{i}_a, \mathcal{R}^{i}_a)$, $g^{i}_a=|V^{i}|$ is the number of level nonterminals and $G^{i}_a = \sum_{X \in V^{i}} |rhs(X)|$ is the level size, with $G^{i}_b$ and $g^{i}_b$ being the corresponding values in $(V^{i}_a, \mathcal{R}^{i}_b)$. Finally, we define $S_a \rightarrow C_a[1...k_a]$ as the rule for the start symbol in $\mathcal{G}_a$ (respectively, $S_b \rightarrow C_b[1..k_b]$ for $\mathcal{G}_b$). %$S_a \rightarrow C_a[1...k_a]$ has grammar level $l_a$, while $S_b \rightarrow C_b[1..k_b]$ has grammar level $l_b$.

\textsc{MergeGrams} processes the grammar levels $1,2,\ldots, \max(l_{a}, l_{b})-1$ in increasing order. In each round $i$, we keep the invariant that the right-hand sides of $\mathcal{R}^{i}_a$ and $\mathcal{R}^{i}_b$ are comparable. That is, given two rules $X_a \rightarrow Q_a[1..q_a] \in \mathcal{R}^{i}_a$ and $X_b \rightarrow Q_b[1..q_b] \in \mathcal{R}^{i}_b$, $Q_a[1..q_a]=Q_b[1..q_b]$ implies $exp(X_a)=exp(X_b) \in \Sigma^{*}$. When $i=1$, the invariant holds as the right-hand sides of $\mathcal{R}^{1}_a$ and $\mathcal{R}^{1}_b$ are over $\Sigma$, which is the same for $\mathcal{T}_a$ and $\mathcal{T}_b$ (Section~\ref{sec:str}). 

We begin the algorithm by creating an array $L_a[0..l_a]$ that stores in $L_a[i] = \sum_{j=0}^{i-1} |V^{j}_a|$ the number of nonterminals with level $<i$. Observe that $L_a[0]=0$ and $L_a[1]=|\Sigma|$ because $V^{0}_a=\Sigma$. We create an equivalent array $L_b$ for $\mathcal{G}_b$. We also initialize the sets $E^1, E^2, \ldots, E^{l_b-1}$ where every $E^i \subseteq [1..k_b]$ keeps the indexes of $C_b[1..k_b]$ with level-$i$ symbols.

The first step of the merge round $i$ is to scan the right-hand sides of $\mathcal{R}^{i}_a$, and for each $X_a \rightarrow Q_a[1..q_a]$, we modify $X_a=X_a-L_a[i]$ and $Q_a[j]=Q_a[j]-L_a[i-1]$, with $j \in [1..q_a]$. The change $X_a=X_a-L_a[i]$ %is out of convenience. In a later stage of the algorithm, if we see $X_a$ in the right-hand sides of $\mathcal{R}^{i+1}_a$, we can unequivocally locate its rule $X_a \rightarrow Q_a[1..q_a]$ by visiting the $rth$ element of $\mathcal{R}^{i}_a$. This arrangement
allows us to append new elements to $V^{i}_a$ and $\mathcal{R}^{i}_a$ while maintaining the validity of the symbols on the right-hand sides of $\mathcal{R}^{i+1}_a$, whose alphabet is $V^{i}_a$. Then, we create a hash table $H_a$ that stores every modified rule $X_a \rightarrow Q[1..q_a] \in \mathcal{R}^{i}_a$ as a key-value pair $(Q_a[1..q_a], X_a)$, and an empty array $M_b[1..|V^{i}_b|]$ to update the right-hand sides of $R^{i+1}_{b}$.

We check which right-hand sides in $\mathcal{R}^{i}_b$ occur as keys in $H_a$. Recall that, when $i=1$, the strings in $\mathcal{R}^{i}_a$ are already comparable to the keys in $H_a$ because they are over $\Sigma$ and the subtraction of $L_a[i-1=0]$ does not change their values. For $i>1$, we make these strings comparable during the previous round $i-1$. If the string $Q_b[1..q_b]$ of a rule $X_b \rightarrow Q_b[1..q_b]$ occurs in $H_a$ as a key, we extract the associated value $X_a$ from the hash table. On the other hand, if $Q_b[1..q_b]$ does not exist in $H_a$, we create a new symbol $X_a = |V^{i}_a|+1$ and set $V^{i}_a= V^{i}_a \cup \{X_a\}$. Subsequently, we record the new rule $X_a \rightarrow Q_b[1..q_b]$ in $\mathcal{R}^{i}_a$ and store $M[X_b-L_{b}[i]] = X_a$. Once we process $\mathcal{R}^{i}_b$, we scan the right-hand sides of $\mathcal{R}^{i+1}_b$ and use $M$ to update their symbols. We also use $E^{i}$ to update each position $j \in E^{i}$ as $C_b[j]=M[C_b[j]-L_b[j]]$. Now, we discard $\mathcal{R}^{i}_b$ and continue to the next merge round $i+1$.

\begin{figure}[t]
\centering
\includegraphics[width=0.9\textwidth]{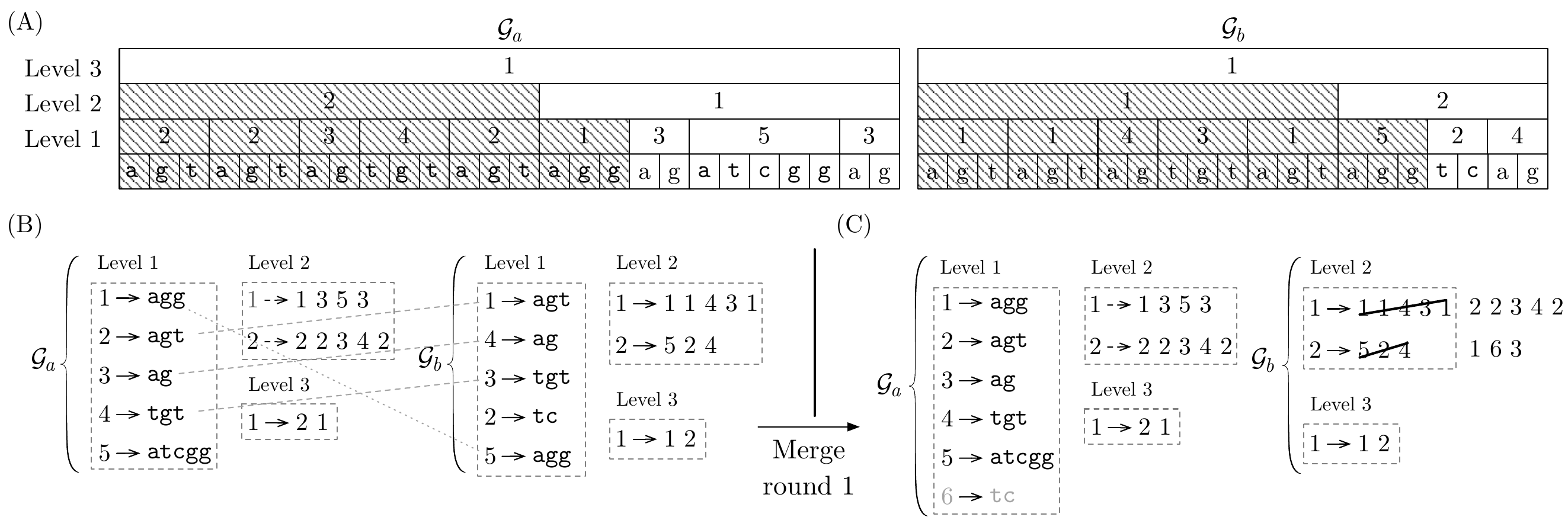}
\caption{
Example of $\textsc{MergeGrams}(\mathcal{G}_a, \mathcal{G}_b)$.
%Example of \textsc{MergeGrams} with two grammars $\mathcal{G}_a$ and $\mathcal{G}_b$ generating the strings $T_a = \texttt{agtagtagtgtagtaggagatcggag}$ and $T_b=\texttt{agtagtagtgtagtaggtcag}$, respectively.
%combining the grammars $\textsc{BuildGram}(T_a, \mathcal{H})=\mathcal{G}_a$ and $\textsc{BuildGram}(T_b, \mathcal{H})=\mathcal{G}_b$.
As the parsing is stable, $T_a[1..19]=T_b[1..19]$ have cores (dashed boxes) with the same topology in $\mathcal{G}_a$ and $\mathcal{G}_b$ . In (B-C), the nonterminals represent the relative position of their rules in their corresponding levels. For example, the left-hand side of $2 \rightarrow 5 2 4$ in $\mathcal{G}_b$ (side B) is $2$ because that rule is the second in level 2. On the other hand, the symbol $2$ in $5 2 4$ refers to the second rule of level 1. In the first merge round, \textsc{MergeGrams} checks which right-hand sides in level one of $\mathcal{G}_b$ are also right-hand sides in level one of $\mathcal{G}_b$ (dashed lines in side B). Only \texttt{tc} is not in $\mathcal{G}_a$, so the algorithm appends it at the end of level one in $\mathcal{G}_a$ and assigns it the new metasymbol $6$ (side C). Subsequently, it discards level one in $\mathcal{G}_b$ and updates the right-hand sides of level two in $\mathcal{G}_b$ according to their corresponding metasymbols in $\mathcal{G}_a$. In (C), the rule $1 \rightarrow 1 1 4 3 1$ becomes $1 \rightarrow 2 2 3 4 2$ and $2 \rightarrow 5 2 4$ becomes $2 \rightarrow 1 6 3$. For example, $5$ becomes $1$ on the right-hand side of $2 \rightarrow 5 2 4$ because the level one rule $5 \rightarrow \texttt{agg}$ in $\mathcal{G}_b$ matches $1 \rightarrow \texttt{agg}$ in $\mathcal{G}_a$ (see dashed lines in side B). After the update, \textsc{MergeGrams} goes to the next round and operates recursively.}
\label{fig:merge_gram}
\end{figure}

When we finish round $i$, the right-hand sides of $\mathcal{R}^{i+1}_b$ are over $[1..|V^{i}_a|]$, and the right-hand sides of $\mathcal{R}^{i+1}_a$ will be over $[1..|V^{i}_a|]$ after we update their values with $L_a$. These modifications will make both strings sets comparable, and our invariant will hold for $i+1$.

After $\min(l_a, l_b)-1$ rounds of merge, one of the input grammars will run out of levels. The remaining rounds will skip the creation and query of $H_a$, and will append new rules directly to $\mathcal{R}^{i}_a$ (if any). After we finish the rounds, we concatenate the compressed strings to form $C_{ab}[1..k_a+k_b] = C_a[1..k_a]{\cdot}C_b[1..k_b]$ and update the starting rule $S_a \rightarrow C_{ab}[1..k_a+k_b]$. The last step of \textsc{MergeGrams} is to modify the symbols in $\mathcal{G}_{a}$. For that purpose, we recompute $L_a$%(the content of each $V^{i}_a$ could have changed after the rounds),
, and for every level-$i$ rule $X_a \rightarrow Q_a[1..q_a]$, we set $X_a = X_a+L_a[i]$ and $Q_a[j]=Q_a[j]+L_a[i-1]$, with $j \in [1..q_a]$. %Finally, we return $\mathcal{G}_{a}$ and we are done.
Figure~\ref{fig:merge_gram} shows an example of \textsc{MergeGrams}. 

\begin{theorem}
Let $\mathcal{G}_a$ (respectively, $\mathcal{G}_b$) be a locally consistent grammar that generates strings in a collection $\mathcal{T}_a$ of $k_a$ strings (respectively, a collection $\mathcal{T}_b$ with $k_b$ strings). The size of $\mathcal{G}_a$ is $G_a$ and the size of $\mathcal{G}_b$ is $G_b$. Similarly, let $g_a=|\Sigma|+|V_a|$ be the number of grammar symbols (respectively, $g_b=|\Sigma|+|V_b|$). {\normalfont $\textsc{MergeGrams}(\mathcal{G}_a, \mathcal{G}_b)=\mathcal{G}_{ab}$} builds a locally consistent grammar generating strings in $\mathcal{T}_{ab}=\mathcal{T}_a \circ \mathcal{T}_b$ in $O(G_a+G_b)$ time w.h.p. and $O(G_a\log g_a + G_b\log g_b)$ bits of space 
\end{theorem}

\begin{proof}
We obtain $L_a$ and $L_b$ in one scan of the nonterminal sets, which takes $O(g_a+g_b)$ time, It is also possible to obtain the sets $E^1, E^2, \ldots, E^{l_b}$ in $O(G_b)$ time. Let $G_a^{i}$ (respectively, $G_b^{i}$) be the number symbols on the right-hand sides of $\mathcal{R}^{i}_a$ (respectively, $\mathcal{R}^{i}_b$). Filling in the hash table $H_a$ requires a linear scan of $\mathcal{R}^{i}_a$, which runs in $O(G^{i}_a)$ time w.h.p. On the other hand, scanning $\mathcal{R}^{i}_b$ and querying its right-hand sides in $H_a$ takes $O(G^{i}_b)$ time w.h.p. In addition, modifying the right-hand sides of $\mathcal{R}^{i+1}_b$ with $M$ takes $O(G^{i+1}_{b})$ time. If we transfer the cost of updating $\mathcal{R}^{i+1}_b$ to the next round $i+1$, performing the merge round $i$ takes $O(G_a^{i}+G_{b}^{i})$ time w.h.p. Now, let $g^{i}_{a}=|V^{i}_a|$ and $g^{i}_b=|V^{i}_b|$ be the number of level-$i$ symbols. We require $G^{i}_a\log g^{i-1}_a + G^{i}_b\log g^{i-1}_b$ bits to encode $\mathcal{R}^{i}_a$ and $\mathcal{R}^{i}_b$, $O(G^{i}_a\log g^{i-1}_a)$ bits for $H_a$, and $g^{i}_b\log g^{i}_b$ bits for $M$. Consequently, the cost of the round $i$ is $O(G^{i}_a\log g^{i-1}_a + G^{i}_b\log g^{i-1}_b)$ bits. As $G_a=k_a+|\Sigma|+\sum_{j=1}^{l_a-1} G^{i}_a$ and $G_b=k_b+|\Sigma|+\sum_{j=1}^{l_b-1} G^{i}_b$, then \textsc{MergeGrams} runs in $O(G_a+G_b)$ time w.h.p and uses $O(G_a\log g_a + G_b\log g_b)$ bits of space.
\end{proof}

\section{Experiments}\label{sec:exp}

\textbf{Implementation details}. We implemented our framework in \texttt{C++} in a tool called \texttt{LCG} (\url{https://github.com/ddiazdom/lcg}). We support parallel compression by interleaving executions of \textsc{BuildGram} and \textsc{MergeGrams} as follows: given a collection $\mathcal{T}$ and an integer $p$, \lcg~uses $p$ threads that execute \textsc{BuildGram} in parallel to compress different subsets of $\mathcal{T}$ into different buffer grammars. When the combined space usage of the buffer grammars exceeds a given threshold, \lcg~merges them into a sink grammar using \textsc{MergeGrams} and resets the buffer grammars. Section~\ref{sec:imp_det} explains this idea in more detail. We refer to this strategy as \textsc{PBuildgram} to differentiate it from our description in Section~\ref{sec:gram_algo}. After running \textsc{PBuildGram}, \lcg~run-length compresses the output (\textsc{RL} step), and then removes unique nonterminals from the output of \textsc{PBuildGram} + \textsc{RL} (\textsc{Simp} step). %In the following, we use MB, GB, and TB to denote space measurements in powers of ten and MiB, GiB, and TiB to denote measurements in powers of two. 

\textbf{Experimental setup and inputs}. We compared \lcg~against other state-of-the-art compressors, measuring the compression speed in MB per second (MB/s), the peak of the working memory in bits per symbol (bps), and the compression ratio. Furthermore, we assessed the amount of compression \lcg~achieves, its resource usage, and how it scales with the number of threads. We conducted the experiments on a machine with AlmaLinux 8.4, 3 TiB of RAM, and processor Intel(R) Xeon(R) CPU E7-8890 v4 @ 2.20GHz, with 192 cores. We tested four collections. \humans: all the human genome assemblies available in NCBI up to August 27, 2024 (3.46 TB, $\sigma$=16). \ATB: release 2.0 of the AllTheBacteria dataset~\cite{atb24}, which contains genomes of bacteria and archaea (7.9 TB, $\sigma=5$). \covid: all the SARS-CoV-2 genomes in NCBI up to November 5 ($267.4$ GB, $\sigma=16$). \kernel: last 40 versions of the Linux kernel (.h, .c, .txt, and .rst files) up to December 13, 2024 ($54.4$ GB, $\sigma=190$). %Table~\ref{tab:datasets} has more information on the datasets and Table~\ref{tab:results} shows the experimental results. 

\textbf{Competitor tools}. \sevz~(\url{https://github.com/facebook/zstd}) is an efficient tool that uses a simplified version of LZ and encodes the output using Huffman~\cite{huff} and ANS~\cite{ans}. \agc~(\url{https://github.com/refresh-bio/agc}) is a compressor for highly similar genomic sequences by Deorowicz et al.~\cite{d23agc} that breaks strings into segments and groups segments into blocks according to sequence similarity. \rp~(\url{https://users.dcc.uchile.cl/~gnavarro/software/repair.tgz}) is a popular grammar compression algorithm by Larsson and Moffat~\cite{l2000off} that recursively replaces the most common pair of symbols in the text. \bigrp~(\url{https://gitlab.com/manzai/bigrepair}) is a RePair variant by Gagie et al.~\cite{g2019rpair} that scales the compression by using prefix-free parsing~\cite{boucher2019prefix}. \lcg, \agc, \sevz, and~\bigrp~support multithreading, so we used 16 threads in each. \rp~does not support multithreading. %, and due to time constraints, we did not run the other tools with one thread.
For \sevz, we used compression level 15 and a window size of 2 GiB --the tool does not allow longer windows with compression level $\geq 15$. 

\subsection{Results and discussion}

\begin{table}[t]
    \centering
    \resizebox{\textwidth}{!}{%
        \begin{tabular}{l|cccc|cccc|cccc}
            \hline
            \multicolumn{1}{c|}{Tool} & \multicolumn{4}{c|}{Compression ratio} & \multicolumn{4}{c|}{Compression speed} & \multicolumn{4}{c}{Working memory}\\
            \multicolumn{1}{c|}{}         & \multicolumn{4}{c|}{plain/compressed}  & \multicolumn{4}{c|}{MB/s}        & \multicolumn{4}{c}{bps}\\
            & \ATB & \humans & \covid & \kernel & \ATB & \humans & \covid & \kernel& \ATB & \humans & \covid & \kernel\\
            \hline
            \lcg   & \textbf{85.26} & 135.54 & 328.10 & 99.99  & \textbf{232.26} & \textbf{244.73} & \textbf{506.44} & \textbf{237.91} & 0.43 & 0.29 & 0.36 & 2.05\\
            \agc   & - & \textbf{144.90} & 237.93 & - & - & 120.42 & 53.71 & - & - & 0.15 & 0.18 & -\\
            \sevz  & 58.19 & 4.72 & 344.99 & 38.23 & 95.51 & 27.21 & 442.79 & 226.45 & \textbf{0.004} & \textbf{0.01} & \textbf{0.10} & \textbf{0.53}\\
            \rp    & - & - & \textbf{778.62} & \textbf{127.03} & - & - & 1.89 & 2.21 & - & - & 30.63 & 150.11 \\
            \bigrp & - & - & 586.12 & 109.87  & - & - & 47.89 & 21.04 & - & - & 2.28 & 3.69\\
            \hline
        \end{tabular}
    }
    \caption{Performance of the competitor tools. The best result of each column is in bold. Cells with a dash are experiments that crashed or are incompatible with the input.}
    \label{tab:results}
\end{table}
\subsubsection{Comparison with other text compressors}

\lcg~was the fastest tool, with a speed ranging $232.26 - 506.44$ MB/s (Table~\ref{tab:results}). The speed of the other tools varied with the input, but \rp~remained the slowest ($1.89$-$2.21$ MB/s). The fact that \rp~used one thread and the other tools 16 alone does not explain these results. For instance, \lcg~was 268 times faster than \rp~in \covid~(506.44 MB/s versus 1.89 MB/s, respectively). We also observed that \lcg~achieved higher speeds in more compressible inputs (e.g. \covid), probably because the hash tables recording phrases from the text (see encoding in Section~\ref{sec:gram_enc}) have to perform more lookups than insertions --lookups are cheaper.

The most space-efficient tool was~\sevz, with a working memory usage of $0.004-0.53$ bps. This result is due to the cap of $2$ GiB that \sevz~uses for the LZ window, regardless of the input size. This threshold keeps memory usage low, but limits compression in large datasets where repetitiveness is spread (\humans~or \kernel). On the other hand, \sevz~with \ATB~yielded important space reductions probably because Hunt et al.~\cite{atb24} preprocessed \ATB~to place similar strings close to each other to improve LZ compression. \lcg~used far less working memory than the other grammar compressors ($0.29-2.05$ bps versus $30.63-150.11$ of \rp~and $2.28-3.69$ of \bigrp), though it is still high compared to \sevz.

\rp~obtained the best compression ratios and is likely to outperform the other tools in \ATB~and \humans~--where we could not run \rp. Compared to \lcg, \rp~achieved $2.37$ times more compression in \covid~and 1.27 times more in~\kernel. The difference was smaller with~\bigrp, with \rp~achieving $1.32$ times more compression in \covid~and $1.15$ times in \kernel. Despite \lcg~did not obtain the best ratio, it still achieves important reductions, and its trade-off between compression and resource usage seems to be the best. Besides, it is still possible to further compress in \lcg~by applying RePair or Huffman encoding. We think these additional steps would be fast, as they operate over a small grammar.

\begin{table}[t]
    \centering
    \resizebox{\textwidth}{!}{%
    \begin{tabular}{
    l
    S[table-format=-1.0e-1, round-precision = 2, round-mode = figures, scientific-notation = true]
    S[table-format=-1.0e-1, round-precision = 2, round-mode = figures, scientific-notation = true]
    S[table-format=-1.0e-1, round-precision = 2, round-mode = figures, scientific-notation = true]
    S[table-format=-1.0e-1, round-precision = 2, round-mode = figures, scientific-notation = true]
    S[table-format=-1.0e-1, round-precision = 2, round-mode = figures, scientific-notation = true]
    S[table-format=-1.0e-1, round-precision = 2, round-mode = figures, scientific-notation = true]
    S[table-format=-1.0e-1, round-precision = 3, round-mode = figures, scientific-notation=true]
    c
    c
    c
    }
        \hline
        \multicolumn{1}{l}{Dataset} & \multicolumn{2}{c}{\texttt{RePair}} &\multicolumn{8}{c}{\texttt{LCG}} \\\cmidrule(lr){2-3}\cmidrule(lr){4-11}
                                     &  &                                  &\multicolumn{4}{c}{\textsc{PBuildGram}}                                                                                    & \multicolumn{2}{c}{\textsc{RL}}      & \multicolumn{2}{c}{\textsc{Simp}}\\\cmidrule(lr){4-7}\cmidrule(lr){8-9}\cmidrule(lr){10-11}
                                     &\multicolumn{1}{c}{$g$} & \multicolumn{1}{c}{$G$}   &\multicolumn{1}{c}{$g$}    & \multicolumn{1}{c}{$G$} & \multicolumn{1}{c}{$max(g^{i})$} & \multicolumn{1}{c}{$max(G^{i})$} & \multicolumn{1}{c}{\% nonter.}   & Size      & \% deleted & Size                   \\
                                     & &                                                    &                           &                         &                                  &                                   & \multicolumn{1}{c}{increase} & reduction & rules   & reduction              \\
        \hline
        \ATB                         &\text{-} & \text{-}  & 9487673329 & 29475041123 & 3289664779 & 10457037296 & 0.00668 & 0.69  & 83.35 & 27.02 \\
        \humans                      &\text{-} & \text{-}  & 2397891525 & 12287192412 & 445803832  & 1454454746  & 0.02161 & 36.73 & 72.30 & 22.31 \\
        \covid                       &16568289 & 114476810 & 108043879  & 507843065   & 15270975   & 180659822   & 0.00522 & 35.16 & 88.71 & 29.11 \\
        \kernel                      &65668740 & 131876293 & 64044303   & 217799801   & 23924433   & 75625197    & 0.02696 & 1.76  & 82.25 & 24.62 \\
        \hline
    \end{tabular}
    }
    \caption{Compression statistics. The value $g=|V|$ is number of nonterminals, $G=\sum_{X \in V} |rhs(X)|$ is the grammar size, $max(g^{i}) = \max_{i \in [1..\ell-1]} |V^{i}|$ is the maximum number of nonterminals in a level, $max(G^{i}) = \max_{i \in [1..l-1]} \sum_{X \in V^{i}} |rhs(X)|$ is the maximum level size. ``Size reduction'' is the percentage of decrease for $G$ relative to previous grammar. The column ``\% nonter. increase'' is the percentage of increase for $g$ relative to the grammar of \textsc{PBuildGram}.}
    \label{tab:recomp_results}
\end{table}

\subsubsection{Breakdown of our method}

\textbf{Compression}. \textsc{PBuildGram} produced substantially larger grammars than \rp, being $69.32\%$ larger in \kernel~and $363.63\%$ larger in \covid~(see Table~\ref{tab:recomp_results}). However, we expected this result as \textsc{PBuildGram} is not greedy. Interestingly, \textsc{PBuildGram} produced fewer nonterminals than \rp~in \kernel~($6.4{\times}10^{7}$ versus $6.6{\times}10^{7}$, respectively). The maximum number of nonterminals $max(g^{i})$ produced in a round was $3.3 \times 10^{9} < 2^{32}-1$ (\ATB), indicating that 32 bit fingerprints (i.e. $m^{i}<2^{32}$ in $\mathcal{H}$) are likely to be enough to keep the number of collisions low in repetitive collections with dozens of TB, although nonrepetitive collections of this size might require larger fingerprints. Recall from Section~\ref{sec:gram_algo_imp} that the more bits we use for fingerprints, the less collisions we have, and hence less impact on the compression, but it also means more working memory. The effect of \textsc{RL} varied with the input: the grammar size reductions ranged between $0.69\%$ and $36.73\%$, with a negligible increase in the number of nonterminals. \textsc{RL} performed well in \humans~and \covid~because they have long equal-symbol runs of \texttt{N}s. On the other hand, \textsc{Simp} removed $81.65\%$ of the nonterminals and reduced the size of the grammar by $25.76\%$, on average. 

\textbf{Resource usage}. The grammar encoding we chose (Section~\ref{sec:gram_enc}) had an important effect on the usage of resources. \lcg~spent $93.2\%$ of its running time, on average, executing~\textsc{PBuildGram} (Figure~\ref{fig:lcg_per}A). The bottleneck was the lookup/insertion of phrases in the hash tables $H^{1}, H^{2}, \ldots, H^{l-1}$ of the grammars when \textsc{BuildGram} (subroutine of \textsc{PBuildGram}) parsed its input text. These hash table operations are costly because of the high number of cache misses and string comparisons. In particular, the first three parsing rounds of \textsc{BuildGram} are the most expensive because the text is not small enough and produces a large set of phrases (see Figure~\ref{fig:pround_breakdown}). Consequently, \textsc{BuildGram} has to hash more frequently and in larger hash tables in those rounds. The impact of the other steps is negligible, with \textsc{RL} and \textsc{Simp} accounting (on average) for $1.82\%$ and $4.98\%$ of the running time, respectively. The working memory usage (Figure~\ref{fig:lcg_per}B) varied with the input: in \ATB~and \humans, the sink grammar and the fingerprints of \textsc{PBuildGram} accounted for $75\%-82\%$ of the usage. The rest of the memory was satellite data: the arrays and grammars of the thread buffers. The results are different in \covid~and \kernel, where the memory usage is dominated by satellite data. The hash tables $H^{1}, H^{2}, \ldots, H^{l-1}$ imposed a considerable space overhead as they store their keys (i.e. parsing phrases) in arrays of 32-bit cells to perform fast lookups (via \texttt{memcmp}). We can reduce the memory cost by using hash tables that store keys using VBytes instead. In this way, the keys still use an integral number of bytes, and we can still use \texttt{memcmp} for lookups. Our preliminary tests (not shown here) suggest that using VBytes in the hash tables (and 32-bit fingerprints) reduces the space of the sink grammar by $39\%$ in \humans, $47\%$ in \covid, and $45\%$ in \kernel. If we consider that the buffer grammars use the same encoding, the change to VBytes could drastically reduce the peak of working memory in \lcg. We can decrease working memory even further by keeping parts of the sink grammar on disk. 

%We implemented \lcg~so that the satellite data use space proportional to an input fraction, with this fraction being smaller as the input becomes larger. Therefore, as \covid~and \kernel~are much smaller \ATB~and \humans, satellite data use more space. %On the other hand, \covid~and \kernel~are more repetitive than \ATB~and \humans, so their sink grammars obtain better compression ratios. A possible way to reduce the memory peak of \lcg~would be to represent the grammars (sink and local) of \textsc{PBuildGram} using a variable-length encoding such as VByte~\textcolor{red}{x}. The current implementation uses a plain format: 4 bytes per grammar symbol.
%Although \lcg~uses much less working memory than other grammar compressors, its performance is still far from \sevz~(see Table~\ref{tab:results}). One way to reduce the gap would be to compress the sink/local grammars in \textsc{PBuildGram} using, for example, VByte encoding. 

\textbf{Effect of parallelism}. The compression speed of \lcg~ in \humans~increased steadily with the threads (Figure~\ref{fig:lcg_per}C), while the working memory peak remained stable in $0.29$ bps until $16$ threads. After that, the peak increased to 0.31 bps with $20$ threads, and $0.33$ with $24$.% The increase could be because of the high number of threads and their redundant information.   

\begin{figure}[t]
\centering
\resizebox{\textwidth}{!}{ %
\begin{tikzpicture}[
   declare function={
    barW=3.5pt; % width of bars
    barShift=barW/2; % bar shift
  }
]
\begin{axis}[
    ybar stacked,
    at={(0,0)},
    anchor=south west,
    width=0.3\textwidth, % Relative width of the second plot
    height=0.3\textwidth,
    enlarge x limits=0.2,
    enlarge y limits=0.1,
    ymin=0, ymax=1, 
    legend style={at={(-0.5,0.5)}, anchor=east, column sep=1ex, align=center, legend cell align=left},
    ylabel={Running time},
    symbolic x coords={ATB, HUM, COVID, KERNEL},
    xtick=data,
    x tick label style={rotate=45,anchor=east, font=\sffamily},
    axis x line*=bottom,
    axis y line*=left,
    name=axis1
    ]
\addplot+[ybar, fill=black, color=black] plot coordinates {
  (ATB,0.91)
  (HUM,0.942) 
  (COVID,0.954)
  (KERNEL,0.924)
};
\addlegendentry{\textsc{PBuildGram}}
\addplot+[ybar, white, pattern=north east lines] plot coordinates {
  (ATB,0.02)
  (HUM,0.014) 
  (COVID,0.015)
  (KERNEL,0.022)
};
\addlegendentry{\textsc{RL}}
\addplot+[ybar, gray!50, fill=gray!50] plot coordinates {
  (ATB,0.07)
  (HUM,0.044)
  (COVID,0.031)
  (KERNEL,0.053)
};
\addlegendentry{\textsc{Simp}}
\end{axis}

\begin{axis}[
    ybar stacked,
    at={(6.5cm,0)},
    anchor=south west,
    width=0.3\textwidth, % Relative width of the second plot
    height=0.3\textwidth,
    enlarge x limits=0.2,
    enlarge y limits=0.1,
    ymin=0, ymax=1, 
    legend style={at={(-0.5,0.5)}, anchor=east, column sep=1ex, align=center, legend cell align=left},
    ylabel={Memory peak},
    symbolic x coords={ATB, HUM, COVID, KERNEL},
    axis x line*=bottom,
    axis y line*=left,
    xtick pos=left,
    x tick label style={rotate=45,anchor=east, font=\sffamily},
    name=axis2,
]
\addplot+[ybar, fill=black, color=black] plot coordinates {
  (ATB,0.64)
  (HUM,0.60) 
  (COVID,0.27)
  (KERNEL,0.13)
  };
\addlegendentry{Sink gram.}

\addplot+[ybar, white, pattern=north east lines] plot coordinates {
  (ATB,0.18)
  (HUM,0.15) 
  (COVID,0.07)
  (KERNEL,0.04)
  };
\addlegendentry{Fps}

\addplot+[ybar, gray!50, fill=gray!50] plot coordinates {
  (ATB,0.18)
  (HUM,0.25)
  (COVID,0.66)
  (KERNEL,0.83)
  };
\addlegendentry{Sat. data}
\end{axis}

\begin{axis}[
    ybar,
    anchor=south west,
    width=0.3\textwidth, % Relative width of the second plot
    height=0.3\textwidth,
    bar width=barW, % added
    bar shift=-barShift, % added
    symbolic x coords={4,8,12,16,20,24},
    axis y line*=left,
    axis x line=none,
    ymin=0, ymax=320,
    ylabel=Speed MB/s,
    xlabel=Number of threads,
    enlarge x limits=0.15,
    enlarge y limits=0.1,
    axis x line*=bottom,
    xtick=data,
    at={(11cm,0)},
 ]
\addplot[draw=black, fill=black]
     coordinates {
        (4, 82.68)
        (8, 146.58)
        (12, 200.94)
        (16, 244.73)
        (20, 282.36)
        (24, 314.92)
    };
\end{axis} % start a new axis for the second data set
\begin{axis}[
    ybar,
    anchor=south west,
    width=0.3\textwidth, % Relative width of the second plot
    height=0.3\textwidth,
    bar width=barW,
    bar shift=barShift,
    symbolic x coords={4,8,12,16,20,24},
    axis y line*=right,
    ymin=0, ymax=0.4,
    ylabel=Mem. peak bps,
    xlabel=Number of threads,
    axis x line*=bottom,
    enlarge x limits=0.15,
    enlarge y limits=0.1,
    yticklabel style={color=gray!85},
    ylabel style={color=gray!70, rotate=180},
    at={(11cm,0)},
    name=axis3,
    tick align=inside
    %x tick label style={rotate=45,anchor=east, font=\sffamily},
]
\addplot[gray!70, fill=gray!70]
coordinates {
    (4, 0.2)
    (8, 0.2)
    (12, 0.2)
    (16, 0.2)
    (20, 0.31)
    (24, 0.33)
    };
\end{axis}
\node [left = 4.1 cm of axis1.north west, anchor=south]  {(A)};
\node [left = 3.5 cm of axis2.north west, anchor=south]  {(B)};
\node [left = 1.15 cm of axis3.north west, anchor=south]  {(C)};
\end{tikzpicture}
}
\caption{Performance of \lcg. (A) Running time breakdown. (B) Memory peak breakdown. ``Fps'' are the fingerprints in \textsc{PBuildGram}, and ``Sat. data'' are arrays and grammars of the buffers (Section~\ref{sec:gram_enc}). (C) Performance of \lcg~in \humans~relative to the number of threads. The left y axis is the compression speed and the right y axis is the memory peak.}
\label{fig:lcg_per}
\end{figure}

\section{Conclusions and further work}

We presented a parallel grammar compressor that processes texts at high speeds while achieving high compression ratios. Our working memory usage is still high compared to popular general-purpose compressors like \sevz, but we can greatly reduce the gap by using VByte encoding or keeping some parts of the grammar on disk. On the other hand, we use substantially less memory than popular grammar compressors. In fact, to our knowledge, \lcg~is the only grammar-based tool that scales to terabytes of data without crashing or exhausting computational resources. Furthermore, our simple strategy captures repetitions from distant parts of the text, making it more robust than other widely spread compression heuristics. Additional reductions in \lcg~are possible by using greedy methods, such as RePair, or statistical compression on top of the output grammar, but it will slow down the analysis of strings.
%, although we do not implement any greedy technique, we achieve substantial space reductions. Future work could explore the integration of greedy methods in the \lcg~grammar to achieve further space reductions.
%Additionally, incorporating layers of statistical encoding, such as arithmetic or Huffman coding, could improve compression performance.
%Another promising direction is leveraging \lcg~as an intermediate step to compute the LZ parse or the RePair grammar.
As mentioned, our goal is not only to compress but also to scale string processing algorithms in massive collections. In the literature, it has been shown that the use of locally consistent grammars can speed up those algorithms~\cite{diaz2023bwt,diaz2023com} but the efficient computation of the grammar remains a bottleneck. We solved that problem in this work. Integrating our scheme with those algorithms could enable the processing of an unprecedented volume of strings.

%In the literature, it has been shown that one can efficiently operate on locally consistent grammars to compute the Burrows-Wheeler transform or maximal exact matches~\cite{diaz2023com}. The main

%\section*{Acknowledgements}
%The authors wish to thank the Finnish Computing Competence Infrastructure (FCCI) for supporting this project with computational and data storage resources.

\bibliography{references}

\appendix

\section{Implementation details}\label{sec:imp_det}

This section presents the details of the implementation of \textsc{PBuildGram}, the compression algorithm we implemented in~\lcg. We use \textsc{BuildGram} (Section~\ref{sec:gram_algo}) and \textsc{MergeGrams} (Section~\ref{sec:mer_gram}) as building blocks to parallelize the process without losing compression power. Section~\ref{sec:gram_enc} describes the encoding we use for locally consistent grammars, Section~\ref{sec:gram_mod} introduces some changes in \textsc{BuildGram} to make parallel execution easier, and Section~\ref{sec:pbuild_gram} describes the steps of \textsc{PBuildGram}.  

\subsection{Grammar encoding}\label{sec:gram_enc}

Let $\mathcal{G}=\{\Sigma, V, \mathcal{R}, S\}$ be a locally consistent grammar of height $l$ generating strings in $\mathcal{T}=\{T_1, T_2, \ldots, T_k\}$. Additionally, let $(\mathcal{R}^{i}, V^{i})$ be the $ith$ pair, with $i \in [1..l-1]$, of the level-based partition of $\mathcal{G}$ (see Section~\ref{sec:mer_gram}). A hash table $H^{i}$ keeps for every $X \rightarrow Q[1..q] \in \mathcal{R}^{i}$ the key-value pair $(Q[1..q], X)$, while an array $F^{i}[1..|V^{i}|]$ stores the fingerprints in $h^{i}$ of $V^{i}$. The array $F^{0}[1..\sigma]$ stores the fingerprints in $h^{0}$ for $\Sigma$. Additionally, the array $C[1..k]$ maintains the symbols representing the compressed strings of $\mathcal{T}$. Similarly to Section~\ref{sec:mg_imp}, the sets $E^{1}, E^{2}, \ldots, E^{l-1}$ store in $E^{i} \subseteq [1..k]$ the indexes of $C[1..k]$ with level-$i$ symbols. Overall, the encoding of $\mathcal{G}$ is $(\Sigma, F^{0}), (H^{1}, F^{1}, E^{1}), \ldots, (H^{l-1}, F^{l-1}, E^{l-1})$, and $C$. We remark that the nonterminals in $(\mathcal{R}^{i}, V^{i})$ represent ranks. Specifically, the value $X=r$ for a rule $X \rightarrow Q[1..q]$ means that this rule is the $rth$ in $\mathcal{R}^{i}$, while $Q[u]=r$, with $u \ [1..q]$, means that the rule where $Q[u] \in V^{i-1}$ is the left-hand side of the $rth$ in $\mathcal{R}^{i-1}$. The use of ranks simplifies the merge with other grammars (see Section~\ref{sec:mg_imp}). From now on, we will use the operator $space(\mathcal{G})$ to denote the amount of bits for the encoding. We extend it so that $space(\mathcal{G}_1, \mathcal{G}_2, \ldots, \mathcal{G}_p)$ represents $\sum^{p}_{j=1} space(\mathcal{G}_j)$. We implemented each hash table $H^{i}$ using Robin Hood hashing, storing the nonterminal phrases in 32-bit cells. Finally, we implemented $\mathcal{H}$ using the \texttt{C++} library \texttt{xxHash}\footnote{\url{https://github.com/Cyan4973/xxHash}}. 

\subsection{Modification to the grammar algorithm}\label{sec:gram_mod}

We modify \textsc{BuildGram} to operate in parallel efficiently. The new signature of the algorithm is $\textsc{BuildGram}(\mathcal{T}, \mathcal{H}, \mathcal{G}_a, \mathcal{G}_b)=\mathcal{G}_a$, where $\mathcal{G}_a$ and $\mathcal{G}_b$ are two (possible nonempty) grammars, and the output is an updated version of $\mathcal{G}_a$. In this variant, we only record phrases in $\mathcal{T}$ that are not in $\mathcal{G}_a$ or $\mathcal{G}_b$. We store these phrases in $\mathcal{G}_a$ and keep $\mathcal{G}_b$ unchanged, using it as a ``read-only'' component. However, the general strategy to compress $\mathcal{T}$ remains the same. The convenience of the change will become evident in the next section. We assume that $\mathcal{G}_a$ and $\mathcal{G}_b$ use the encoding we described in Section~\ref{sec:gram_enc}. We add the subscript $a$ or $b$ to the encoding's components to differentiate their origin, $\mathcal{G}_a$ or $\mathcal{G}_b$ (respectively).

\textsc{BuildGram} now works as follows: let us assume that we are in the $i$th parsing round and that we receive the partially compressed collection $\mathcal{T}^{i}$ as input. The alphabet of each $T^{i}_j \in \mathcal{T}^{i}$ is $V^{i-1}_a \cup V^{i-1}_b$, and we use the fingerprints in $F^{i-1}_a$ and $F^{i-1}_b$ to compute the types of each $T^{i}[\ell]$ (Equation~\ref{eq:lex_types}) and thus the parsing phrases. Notice that, when $i=1$, it holds $V^{0}_a=V^{0}_b=\Sigma$ and $F^{i-1}_a=F^{i-1}_b$. Let $Q[1..q]=T^{i}_j[\ell..\ell+q-1]$ be the active phrase in the parsing. We first check if $Q[1..q]$ exists in $H^{i}_b$ as a key. If that is the case, we get its corresponding value $X$ in the hash table and assign it to the active phrase. On the other hand, if $Q[1..q]$ is not a key in $H^{i}_b$, we perform a lookup operation in $H^{i}_a$. Like before, if $Q[1..q]$ is a key there, we get the corresponding value $X$ and assign it to the active phrase. Finally, if $Q[1..q]$ is not a key in $H^{i}_a$ either, we insert the active phrase $T^{i}_j[\ell..\ell+q-1]=Q[1..q]$ into $H^{i}_a$ associated with a new metasymbol $X$. Let $s_a$ and $s_b$ be the sizes of $H^{i}_a$ and $H^{i}_b$ (respectively) before starting the parsing round $i$. The value we assign to the new metasymbol is $X=s_a+s_b+1$. Subsequently, we replace $Q[1..q]$ by $X$ in $T^{i}_j$ and move to the next phrase. Before ending the parsing round, we store in $F^{i}_a$ the fingerprints in $h^{i}$ for the new phrases we inserted in $H^{i}_a$. Let $Q[1..q]$ be one of the new phrases, and let $X$ be its metasymbol. We compute the fingerprint $h^{i}(Q[1..q])$ using Equation~\ref{eq:fast_hash} and store the result in $F^{i}_a[X-s_b]$. Identifying the precedence of a symbol $T^{i+1}[\ell] \in V^{i}$ in the next round $i+1$ is simple: $T^{i+1}[\ell]=Y>s_b$ means that the phrase associated with $Y$ is a key in $H^{i}_a$ and its fingerprint in $h^{i}$ is in $F^{i}_a[Y-s_b]$ ($s_b$ never changes). On the other hand, $Y \leq s_b$ means that its phrase is in $H^{i}_b$ and its fingerprint in $h^{i}$ in $F^{i}_b[Y]$. Recall that we need $F^{i}_a$ and $F^{i}_b$ to compute the position types for $T^{i+1}_x$. Similarly, we can obtain the fingerprints in $h^{i+1}$ for the new phrases of $H^{i}_a$ by modifying Equation~\ref{eq:fast_hash} to receive two fingerprint arrays $F^{i}_a$ and $F^{i}_b$ instead of one.

\subsection{Parallel grammar construction}\label{sec:pbuild_gram}

Now that we have explained our grammar encoding and the changes that we need to compress in parallel, we are ready to describe \textsc{PBuildGram}. Figure~\ref{fig:pbuild_gram} shows the steps in detail. We receive as input a string collection $\mathcal{T}$ (a file), the number $p$ of compression threads, and a threshold $t$ indicating the approximate amount of working memory we can use. We first initialize $p$ distinct buffers $(B_1, \mathcal{G}_1, K_1), (B_2, \mathcal{G}_2, K_2), \ldots (B_p, \mathcal{G}_p, K_p)$ and an empty grammar $\mathcal{G}$. Each buffer $j$ is a triplet $(B_j, \mathcal{G}_j, K_j)$, where $B_j$ is an array to store chunks of $\mathcal{T}$ (one at a time), $\mathcal{G}_j$ is the grammar where we compress the chunks we load into $B_j$, and $K_j$ is an array of pairs $(x_1,y_1), (x_2, y_2)\ldots$ indicating the chunks we have loaded into $B_j$. Specifically, $K_j[u]=(x, y)$ means that the $uth$ chunk we loaded into $B_j$ contained the subset $T_{x}, T_{{x}+1}, \ldots, T_{x+y-1} \subseteq \mathcal{T}$. Notice that these strings lie contiguously in $\mathcal{T}$'s file. There is also an array $K$ with information equivalent to that of $K_j$ but for $\mathcal{G}$. %The total length $\sum^{x+y-1}_{k=x} |T_k|$ in bytes of the chunk is, at most, 200 MB. This threshold is out of convenience \textcolor{red}{...}. 

\begin{figure}[t]
\centering
\includegraphics[width=0.8\textwidth]{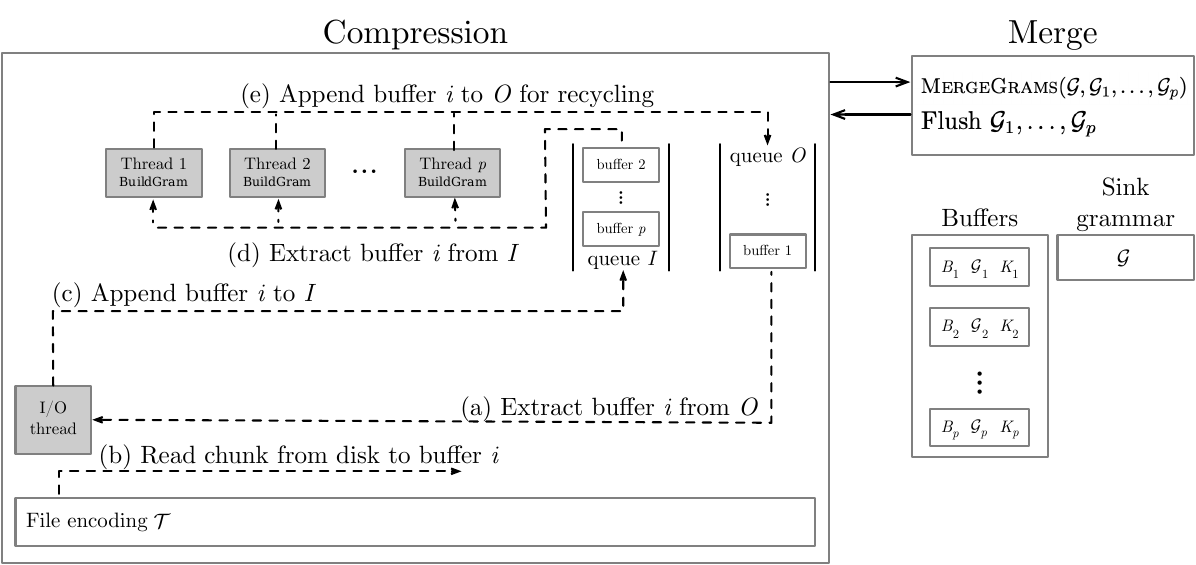}
\caption{Schematic representation of \textsc{PBuildGram}. The steps (a-e) indicate the cycle of a buffer during the compression step.}
\label{fig:pbuild_gram}
\end{figure}

\textsc{PBuildGram} consists in a loop that interleaves two steps, compression and merge. During the compression step, $p$ parallel threads compress the buffers $B_1, B_2, \ldots, B_p$ into the corresponding buffer grammars $\mathcal{G}_1, \mathcal{G}_2, \ldots, \mathcal{G}_p$, and continue doing it while $space(\mathcal{G}_1, \mathcal{G}_2, \ldots, \mathcal{G}_p) < t$. When the space exceeds the threshold, a merge step collapses $\mathcal{G}_1, \mathcal{G}_2, \ldots, \mathcal{G}_p$ into the sink grammar $\mathcal{G}$ and flushes the $p$ buffers. The algorithm then enters a new iteration that restarts the compression from the point in $\mathcal{T}$ where it left it the last time.

The compression step initializes an I/O thread that reads $\mathcal{T}$ from the disk, loading the chunks sequentially from left to right in the arrays $B_1, B_2, \ldots, B_p$. On the other hand, the $p$ compression threads process the arrays $B_1, B_2, \ldots, B_p$ in parallel using the variant of $\textsc{BuildGram}$ we described in Section~\ref{sec:gram_mod}. Every thread receives a buffer $j \in [1..p]$ and runs $\textsc{BuildGram}(B_j, \mathcal{H}, \mathcal{G}_j, \mathcal{G})=\mathcal{G}_j$. We syncronize the I/O thread with the compression threads using two concurrent queues $I$ and $O$. The queue $I$ keeps the chunks that are ready to be processed by the compressor threads, whereas $O$ contains the buffers that were already processed and can be recycled by the I/O thread to insert new chunks. When the algorithm starts, $O$ contains all buffers. 

The synchronization process works as follows: let $u$ be the next chunk of $\mathcal{T}$ that \textsc{PBuildGram} has to process. The I/O thread extracts the head $(B_j, \mathcal{G}_j, K_j)$ from $O$, reads the $u$th chunk $T_{x}, T_{x+1}, \ldots, T_{x+y} \subseteq \mathcal{T}$ from disk, and loads it into $B_j$. Subsequently, it appends $(x,y)$ to $K_j$, and finally it appends $(B_j, \mathcal{G}_j, K_j)$ to $I$. The I/O thread continues to process the next chunks $u+1, u+2, \ldots$ in the same way as long as the compression process remains active. On the other hand, each compression thread tries to acquire the next buffer available $(B_j, \mathcal{G}_j, K_j)$ from the head of $I$. After the thread acquires the buffer and runs \textsc{BuildGram}, it flushes $B_j$ and pushes $(B_j, \mathcal{G}_j, K_j)$ into $O$, thus marking this buffer for recycling. Notice that a compression thread can process multiple noncontiguous chunks of $\mathcal{T}$ and collapse their information into the same grammar $\mathcal{G}_j$. However, later in the execution of the algorithm, we modify $\mathcal{G}_j$ using the information in $K_j$ to fix this problem. 

During the merge step, we execute $\textsc{MergeGram}(\mathcal{G}, \mathcal{G}_j) = \mathcal{G}$ with each $j \in [1..p]$. It is possible to collapse the grammars in parallel in a merge sort fashion: let us assume w.l.o.g that $p+1$ is a power of two. Thus, $(p+1)/2$ threads execute in parallel processes $\textsc{MergeGram}(\mathcal{G}_1, \mathcal{G}_2)=\mathcal{G}_1, \textsc{MergeGram}(\mathcal{G}_3, \mathcal{G}_4)=\mathcal{G}_3, \ldots, \textsc{MergeGrams}(\mathcal{G}, \mathcal{G}_p)=\mathcal{G}$ to produce the new grammars $\mathcal{G}_1, \mathcal{G}_3, \ldots, \mathcal{G}$. Subsequently, the $(p+1)/4$ threads collapse the new grammars in the same way, and the process continues until only one $\mathcal{G}$ remains. Every time we execute $\textsc{MergeGram}(\mathcal{G}_j, \mathcal{G}_{j+1})=\mathcal{G}_j$, we also concatenate the corresponding arrays $K_j$ and $K_{j+1}$ to keep track of the chunks of $\mathcal{T}$ that the resulting $\mathcal{G}_{j}$ encodes. Notice that after the merge, $K$ has all the information of $K_1, K_2, \ldots, K_p$. Finally, we reset the buffers and begin a new iteration of compression.

After we process all the chunks, we perform one last merge step to collapse the buffers in $\mathcal{G}$ and then use the information in $K$ to reorder the elements of $C[1..k]$. Once we finish, we return $\mathcal{G}$ and complete the execution of \textsc{PBuildGram}.

\subsection{Storing the final grammar}

As mentioned, we post-process the output of \textsc{PBuildGram} using \textsc{RL} and then \textsc{Simp}. The resulting file of this process (i.e., the output of~\lcg) uses $G\log g$ bits to encode $\mathcal{R}$,  $g\log G$ bits to store pointers on the right-hand sides of $\mathcal{R}$, and $k \log w$ bits to store pointers to the compressed sequences of the strings in $\mathcal{T}$.

\subsection{Advantage of our parallel scheme}

Our parallel scheme can use a high number of threads with little contention (and thus achieve high compression speeds), while keeping the amount of working memory manageable. A thread executing $\textsc{BuildGram}(B_j, \mathcal{H}, \mathcal{G}_j, \mathcal{G})$  is the only one modifying $B_j$ and $\mathcal{G}_j$, and although the sink $\mathcal{G}$ can be accessed by other threads concurrently, they only read information (i.e., little to no contention). On the other hand, there is some contention when the threads modify the queues $I$ and $O$ concurrently to remove or insert buffers (respectively). However, the compression threads spend most of their time executing \textsc{BuildGram} (modifying a queue is cheap), and it is unlikely that many of them attempt to access the same queue at the same time. On the other hand, the I/O thread might have more contention as we increase the number of threads because it has to compete to modify $O$.%, but we can fix this issue by having more queues where the threads can insert/remove the buffers.

As mentioned above, \textsc{PBuildGram} keeps the consumption of working memory manageable as we add more threads. Specifically, $space(\mathcal{G}_1, \mathcal{G}_2, \ldots, \mathcal{G}_p)$ is not bigger than $space(\mathcal{G})$ by a factor of $p$. In the first compression iteration, the sink grammar $\mathcal{G}$ is empty and because the compression threads do not synchronize when they execute \textsc{BuildGram}, the grammars $\mathcal{G}_1, \mathcal{G}_2, \ldots, \mathcal{G}_p$ will be redundant. Therefore, memory usage will grow rapidly at the beginning, exceeding the memory threshold $t$ and trigger a merge. In this phase, we will collapse redundant content into $\mathcal{G}$ and delete buffer grammars, thus reducing memory consumption. In the next compression iteration, $\mathcal{G}$ will be non-empty, so every instance $\textsc{BuildGram}(B_j, \mathcal{H}, \mathcal{G}_j, \mathcal{G})=\mathcal{G}_j$ will only add to $\mathcal{G}_j$ what is not in $\mathcal{G}$. In addition, $\mathcal{G}_j$ will grow slower with every new merge iteration because there will be less ``new'' sequence content in $\mathcal{T}$. We note that memory usage still depends on $max(t, space(\mathcal{G}))$, with $space(\mathcal{G})$ depending, in turn, on the amount of repetitiveness in $\mathcal{T}$.

\section{Speeding up string processing algorithms (sketch)}\label{sec:speed_up}

In this section, we briefly explain how locally consistent grammars help speed up string processing algorithms (the idea might vary with the application). The most expensive operation in string algorithms is to find matching sequences. We use the fact that a grammar collapses redundant sequence information, so the search space in which a string algorithm has to operate is significantly smaller in the grammar than in the text.

Let $\mathcal{G}=\{\Sigma, V, \mathcal{R}, S\}$ be a locally consistent grammar generating elements in $\mathcal{T}$, and let $\textsc{ALG}(\mathcal{G})$ be a string processing algorithm. As before, we divide $\mathcal{R}=\{\mathcal{R}^{1}, \mathcal{R}^{2}, \ldots, \mathcal{R}^{l-1}\}$ according to levels. We regard the right-hand sides of $\mathcal{R}^{i}$ as a string collection and run (a section of) \textsc{ALG} using $\mathcal{R}^{i}$ as input. The process will give some complete answers and some partial answers. We pass the complete answers to the next level $i+1$ and use them as satellite data. When \textsc{ALG} is recursive and returns from the recursion to level $i$, we use the new information to complete what is missing in the partial answers.

An example of this idea is the computation of maximal exact matches (MEM). Let $lcp(X, Y)$ be the longest common prefix between $exp(X)$ and $exp(Y)$, with $X,Y \in V^{i-1}$. Similarly, let $lcs(X, Y)$ be the longest common suffix of $exp(X)$ and $exp(Y)$. Assume that we have run a standard algorithm to compute the MEMs on the right-hand sides of $\mathcal{R}^{i}$ and that we found a match $A[j_1..j_2]=B[k_1..k_2]$ between two right-hand sides $A, B$ of $\mathcal{R}^{i}$. Completing the MEM requires us to obtain $lcs(A[j_1-1], B[k_1])$ and $lcp(A[j_2+1], B[k_2+1])$. Once we compute all the matches in $\mathcal{R}^{i}$, we use the output so that we can compute $lcs$ and $lcs$ at the next level $i+1$. We still have to project the MEMs to text positions, but this operation is cheaper than computing MEMs in their plain locations, especially if the text is redundant.

The grammar we produce with \textsc{PBuildGram} still requires some modifications to run string algorithms as described above, but the cost of this transformation is proportional to the grammar size, not the text size.

\section{Figures}

\definecolor{gray1}{gray}{0.2}  % Dark gray
\definecolor{gray2}{gray}{0.4}  % Medium-dark gray
\definecolor{gray3}{gray}{0.6}  % Medium-light gray
\definecolor{gray4}{gray}{0.8}

\begin{figure}[!htp]
    \centering
    \resizebox{\textwidth}{!}{
    % First Scatter Plot
    \begin{subfigure}{0.63\textwidth}
        \centering
        \begin{tikzpicture}
            \begin{axis}[
                xlabel={Parsing round},
                ylabel={Nonter. percentage},
                xlabel style={font=\Large},
                ylabel style={font=\Large},
                legend pos=north west,
                name=axisl
            ]
                \addplot[mark=square*, black, thick] 
                coordinates {
                 (1, 0.004)
                 (2, 10.841)
                 (3, 34.673)
                 (4, 19.056)
                 (5, 10.599)
                 (6, 6.486)
                 (7, 4.322)
                 (8, 3.183)
                 (9, 2.537)
                 (10, 2.113)
                 (11, 1.889)
                 (12, 1.821)
                 (13, 1.700)
                 (14, 0.758)
                 (15, 0.018)
                 (16, 0.000)
                };
                %\addlegendentry{ATB}

                \addplot[mark=x, gray1, thick, mark size=5pt] 
                coordinates {
                   (1,0.06)
                   (2,8.07)
                   (3,18.59)
                   (4,15.58)
                   (5,14.08)
                   (6,13.40)
                   (7,11.81)
                   (8,8.68)
                   (9,4.77)
                   (10,2.03)
                   (11,0.94)
                   (12,0.57)
                   (13,0.41)
                   (14,0.37)
                   (15,0.32)
                   (16,0.23)
                   (17,0.09)
                   (18,0.02)
                   (19,0.00) 
                };
                %\addlegendentry{HUM}
                
                \addplot[mark=triangle*, gray2, mark size=4pt] 
                coordinates {
                  (1, 0.60)
                  (2, 6.72)
                  (3, 8.93)
                  (4, 9.18)
                  (5, 9.99)
                  (6, 11.95)
                  (7, 14.13)
                  (8, 13.87)
                  (9, 9.95)
                  (10, 7.37)
                  (11, 7.30)
                };
                %\addlegendentry{COVID}
                \addplot[mark=diamond, gray3, thick, mark size=4pt] 
                coordinates {
                 (1,2.38)
                 (2,31.95)
                 (3,37.36)
                 (4,17.56)
                 (5,6.40)
                 (6,2.29)
                 (7,0.86)
                 (8,0.36)
                 (9,0.19)
                 (10,0.14)
                 (11,0.13)
                 (12,0.12)
                 (13,0.12)
                 (14,0.11)
                 (15,0.03)
                 (16,0.00)
                };
                %\addlegendentry{KERNEL}
            \end{axis}
            \node [left = 0.95 cm of axisl.north west, anchor=south]  {\Large (A)};
        \end{tikzpicture}
    \end{subfigure}
    \hfill
    % Second Scatter Plot
    \begin{subfigure}{0.63\textwidth}
        \centering
        \begin{tikzpicture}
            \begin{axis}[
                xlabel={Parsing round},
                ylabel={Gram. size percentage},
                xlabel style={font=\Large},
                ylabel style={font=\Large},
                legend pos=north east,
                name=axisr,
                legend style={column sep=1ex, legend cell align=left, font=\sffamily}
            ]
            
                \addplot[mark=square*, black, thick] 
                coordinates {
                   (1,0.047)
                   (2,16.919)
                   (3,35.478)
                   (4,18.923)
                   (5,10.399)
                   (6,6.090)
                   (7,3.725)
                   (8,2.433)
                   (9,1.645)
                   (10,1.087)
                   (11,0.754)
                   (12,0.620)
                   (13,0.551)
                   (14,0.244)
                   (15,0.006)
                   (16,0.000)
                };
                \addlegendentry{ATB}

                \addplot[mark=x, gray1, thick, mark size=5pt] 
                coordinates {
                    (1,35.81)
                    (2, 7.54)
                    (3, 11.84)
                    (4, 9.89) 
                    (5, 8.99)
                    (6, 8.53)
                    (7, 7.38)
                    (8, 5.24)
                    (9, 2.70)
                    (10, 1.04)
                    (11, 0.40)
                    (12, 0.19)
                    (13, 0.11)
                    (14, 0.08)
                    (15, 0.06)
                    (16, 0.04)
                    (17, 0.02)
                    (18, 0.00)
                    (19, 0.00)
                };
                \addlegendentry{HUM}
                
                \addplot[mark=triangle*, gray2, mark size=4pt] 
                coordinates {
                  (1,35.57)
                  (2,5.18)
                  (3,6.39)
                  (4,6.60) 
                  (5,6.99)
                  (6,8.49)
                  (7,9.85)
                  (8,9.24)
                  (9,5.62)
                  (10,2.70)
                  (11,1.59)
                };
                \addlegendentry{COVID}
                \addplot[mark=diamond, gray3, thick, mark size=4pt] 
                coordinates {
                (1,4.86)
                (2,34.61)
                (3,34.72)
                (4,15.62)
                (5,5.65)
                (6,2.00)
                (7,0.72)
                (8,0.28)
                (9,0.12)
                (10,0.06)
                (11,0.04)
                (12,0.04)
                (13,0.04)
                (14,0.03)
                (15,0.01)
                (16,0.00)
                };
                \addlegendentry{KERNEL}
            \end{axis}
        %\node [left = 4.1 cm of axis1.north west, anchor=south]  {(A)};
        \node [left = 0.95 cm of axisr.north west, anchor=south]  {\Large (B)};
        \end{tikzpicture}
    \end{subfigure}
    }
    \caption{Number of nonterminals and phrases generated in each parsing round for the output grammar $\mathcal{G}=\{\Sigma, V, \mathcal{R}, S\}$ of \textsc{PBuildGram}. The x-axes are the parsing rounds $i =1,2, \ldots, l$. (A) Percentage $g^{i}/g \times 100$ that the number $g^{i}=|V^{i}|$ of level-$i$ nonterminals contributes to the total number $g=|V|$ of nonterminals. (B) Percentage $G^{i}/G \times 100$ that the number of symbols $G^{i} = \sum_{X \in V^{i}} |rhs(X)|$ contributes to the total grammar size $G$. High percentages denote expensive rounds in each collection.}
    \label{fig:pround_breakdown}
\end{figure}

\end{document}